\documentclass[a4paper,10pt]{article}

\title{A Characterization of Semi-Synchrony for Asynchronous Robots with Limited
Visibility, and its Application to Luminous Synchronizer Design} 

\author{Paola Flocchini\thanks{University of Ottawa, Canada. E-mail: \texttt{paola.flocchini@uottawa.ca}} \and 
Nicola Santoro\thanks{Carleton University, Canada. E-mail. \texttt{santoro@scs.carleton.ca}} \and 
Masafumi Yamashita\thanks{Kyushu University, Japan. E-mail. \texttt{mak@inf.kyushu-u.ac.jp}} \and 
Yukiko Yamauchi\thanks{Kyushu University, Japan. E-mail. \texttt{yamauchi@inf.kyushu-u.ac.jp}} 
}

\usepackage{amsmath,amsfonts} 
\usepackage[pdftex]{graphicx}
\usepackage{latexsym}

\newtheorem{theorem}{Theorem}
\newtheorem{lemma}[theorem]{Lemma}
\newtheorem{corollary}{Corollary}
\newtheorem{assumption}{Assumption} 
\newtheorem{proposition}{Proposition}
\newtheorem{ex}{Example} 

\newenvironment{proof}{{\bf Proof. } } 

\newcommand{\qed}{\hfill $\Box$}

\newcommand{\FSYNC}{\cal{FSYNC}}
\newcommand{\SSYNC}{\cal{SSYNC}}
\newcommand{\ASYNC}{\cal{ASYNC}}

\begin{document}
\date{}
\maketitle

\begin{abstract}
A mobile robot system consists of anonymous mobile robots, each of which 
autonomously performs sensing, computation, and movement
according to a common algorithm,
so that the robots collectively achieve a given task.
There are two main  models of time and activation of the robots.
In the {\em semi-synchronous model} (SSYNC),
the robots share a common notion of time; at each time unit, a subset of 
the robots is
activated,  and each performs all three actions (sensing, computation, 
and movement) in that time unit.
In the {\em asynchronous model} (ASYNC), there is no common notion of time,
the robots are activated at arbitrary times, and the duration of each 
action is arbitrary  but finite.

In this paper, we investigate  the problem of synchronizing  ASNYC 
robots with limited sensing range, i.e., limited visibility. We first 
present a sufficient condition for an ASYNC execution of a common 
algorithm $\mathcal{A}$ to have a corresponding SSYNC execution of $\mathcal{A}$; 
our condition 
imposes timing constraints on the activation schedule of the robots 
and visibility constraints during movement. 
Then, we prove that this condition is necessary (with probability $1$) 
under a randomized ASYNC  adversary.
Finally, we present a synchronization algorithm for luminous ASYNC 
robots
with limited visibility, each equipped with a light that can take a 
constant number of colors.
Our algorithm enables  luminous ASYNC robots
to simulate any algorithm $\mathcal{A}$, designed for the (non-luminous) SSYNC robots 
and satisfying visibility constraints. 

\end{abstract}

\section{Introduction} 

Distributed computing by mobile computing entities has 
attracted much attention in the past two decades 
and many distributed system models have been considered, for example, 
the {\em autonomous mobile robot system}~\cite{SY99} modeled after 
cheap hardware robots with very weak capabilities, 
the {\em population protocol model}~\cite{AADFP06} 
motivated by delay tolerant networks, 
the {\em programmable particle model}~\cite{DDGRSS14} 
inspired by movement of amoebae, and 
the {\em metamorphic robotic system}~\cite{DSY04a,DSY04b} 
considering modular robots. 
The computational power of these distributed systems 
has been investigated in distributed computing theory and 
many fundamental problems have been proposed that require  a degree of 
agreement among the mobile computing entities. 
Typical problems are  {\em leader election}~\cite{DFSBRS15,DPV10}, 
which  requires the entities to 
agree on a single entity,  
{\em gathering}~\cite{CFPS12,FPSW05,SY99}, which  
requires the entities to gather at a point not known apriori, and
{\em shape formation}~\cite{DFSVY20,DSY04a,FYOKY15,SY99,YS10,YY13}
(also called the transformability problem), which 
requires the entities to form a specified shape. 
These results are considered as theoretical foundations 
in several related areas like 
ad-hoc networks, sensor networks, robotics, 
molecular computing, chemical reaction circuits, and so on. 

In this paper, we focus on the autonomous mobile robot system. 
Let $\mathcal{R} = \{r_1, r_2, \ldots, r_n\}$ be a set of $n$ robots. 
Each robot is an {\em anonymous} (indistinguishable) point 
moving in the 2D space. 
The robots are {\em silent} (communication-less) 
and do not have   access to a global coordinate system. 
A robot's  behavior is a repetition of    {\em Look-Compute-Move} cycles:
 in the {\em Look} phase,   it observes the positions of the other robots within 
 its visibility range; in the {\em Compute} phase, it 
computes its next position and a continuous route to the next position 
with a common deterministic algorithm; in the {\em Move} phase, it   moves to the computed position along the computed route. 
The essential properties of mobile robot systems are 
the visibility range, obliviousness, and the timing and activation models. 
A robot $r_i$ is equipped with its own local coordinate system $\mathcal{Z}_i$, 
which is a right-handed $x$-$y$ coordinate system.  
The origin of $\mathcal{Z}_i $ is always the current position of $r_i$, 
while the unit distance and the directions and orientations of 
the $x$ and $y$ axes are arbitrary. 
The observation at $r_i$ is a snapshot in $\mathcal{Z}_i$ containing no additional 
information other than the positions of the robots. 
If the visibility range of a robot is unlimited, 
it can observe all robots, 
otherwise it can observe the robots within its visibility. 
In {\em Compute}, when the input to the common algorithm is the snapshot 
taken in the preceding Look, we say the robots are {\em oblivious}. 
When the input includes past observations and computations, 
we say the robots are {\em non-oblivious}. 
In {\em Move}, 
the movement of a robot is {\em rigid} when the robot always reaches the 
next position, and {\em non-rigid} 
when the robot may stop en route after moving a minimum distance $\delta$ (in $\mathcal{Z}_0$)
(if the  length of the route to the destination is smaller than $\delta$, the destination is reached). 

Three different types of timing and activation models have been proposed: 
In the {\em fully-synchronous model} (FSYNC), 
at each discrete time $t=0, 1, 2, \ldots$, 
all robots execute a Look-Compute-Move cycle 
with each of the Look, Compute, and Move completely synchronized. 
In the {\em semi-synchronous model} (SSYNC), 
at each discrete time $t=0, 1, 2, \ldots$, 
a non-empty subset of robots are activated and execute a Look-Compute-Move cycle 
with each of the Look, Compute, and Move completely synchronized. 
For fairness, we assume that each robot executes infinitely many cycles. 
In the {\em asynchronous model} (ASYNC), 
the robots do not have a common notion of time and 
the length of each cycle is arbitrary but finite. 
We also assume fairness in   ASYNC. 
The main difference between   SSYNC (thus, FSYNC) and 
  ASYNC is that 
in   SSYNC, all robots simultaneously take  a snapshot in Look, 
while in   ASYNC a robot may observe moving robots 
although the robot cannot recognize which robots are moving. 

The effect of obliviousness, asynchrony, and visibility 
on the computational power of autonomous mobile robot systems has been 
extensively investigated~\cite{DPV10,FPSW08,FYOKY15,SY99,YS10}. 
Since the only 
output by the oblivious robots is their geometric positions, 
a fundamental problem is the {\em pattern formation problem}, that 
requires the robots to form a target pattern from an initial configuration. 
Existing literature~\cite{FYOKY15,SY99,YS10}\footnote{
An erratum of \cite{FYOKY15} is available at \cite{FYOKY17}.} 
showed that the initial symmetry among the 
anonymous robots determines the set of formable patterns,
irrespective of obliviousness and asynchrony. The only exception is  
  the point formation problem of two robots, 
  also called the {\em rendezvous problem}. 
In fact, Suzuki and Yamashita have shown  that 
the rendezvous problem is solved by oblivious FSYNC robots, 
but cannot be solved by oblivious SSYNC robots~\cite{SY99}. 
In other words, the rendezvous problem demonstrates the difference between 
  FSYNC and   SSYNC (thus, ASYNC). 
These results consider the robots with unlimited visibility. 
Yamauchi and Yamashita have shown that 
limited visibility substantially shrinks the set of formable patterns 
by oblivious ASYNC robots because the robots do not know their global 
symmetry~\cite{YY13}. 

The robots can overcome the limits 
by distributed coordination or additional capabilities. 
Di Luna et al. have shown that a constant number of oblivious ASYNC robots 
can simulate a single non-oblivious ASYNC robot by encoding the memory contents 
to the geometric positions of the robots~\cite{DFSV18}. 
Das et al. have shown that oblivious ASYNC robots {\em with lights} can 
simulate oblivious SSYNC robots~\cite{DFPSY16}. 
A {\em luminous robot}
is equipped with a light whose color is 
changed in every Look-Compute-Move cycle at the end of Compute 
and observed by other robots. 
The authors showed that luminous ASYNC robots with a constant 
number of colors can simulate an algorithm $\mathcal{A}$ designed 
for oblivious SSYNC robots. 
They presented a {\em synchronizer} 
that makes an activated robot accept or reject the current cycle 
so that the snapshot of an accepted cycle does not contain any moving robot. 
In an accepted cycle, the robot changes the color of its light to ``moving''
and moves to the next position computed by $\mathcal{A}$. 
The synchronizer guarantees   fairness by making all robots wait 
with the ``waiting'' color after it accepts a cycle 
until all the other robots accept a cycle. 
All these techniques are heavily based on the fact that   robots have unlimited visibility. 

In this paper, we investigate synchronization by   oblivious ASYNC robots 
with limited visibility and we make some fundamental contributions. 

We start with a formal definition of simulation by mobile robots. 
A {\em configuration} is the set of positions of the robots in $\mathcal{Z}_0$ 
and an {\em execution} of algorithm $\mathcal{A}$ from an initial configuration $I$ 
is an infinite sequence of configurations. 
In   SSYNC, an execution is a sequence of configurations 
$C_0(=I), C_1, C_2, \ldots$, where $C_t$ is the configuration at time $t$. 
An ASYNC execution is the sequence of configurations 
 $C_{t_0}(=I),  C_{t_1}, C_{t_2}, \ldots$  where at least one robot takes a snapshot, 
with  $t_i < t_{i+1}$  for all $i=1,2, \ldots$. 
Then, the {\em footprint} of a robot is 
the sequence of the positions of the robot in each configuration. 
We say that two executions $E$ and $E'$ 
(possibly in different timing and activation models) 
are {\em similar} when the footprints 
and local observations at the robots 
are identical. 

We then  present a sufficient condition for an ASYNC execution 
to have a similar SSYNC execution, and 
we also show that the condition is necessary with probability $1$ 
under a randomized ASYNC adversary. 
The  randomized impossibility result is novel, based on a Borel 
probability measure space for 
non-rigid movement and asynchronous observations, 
and it provides a stronger argument  than a worst-case (deterministic) analysis.  
Our condition not only requires snapshots of static robots 
but also considers a chain of concurrent observations, 
that cannot be treated separately in a SSYNC execution. 
The transitive closure with respect to the concurrent observations 
forms an equivalence relation 
and cycles  of the ASYNC execution  
are decomposed into equivalence classes. 
We then introduce a ``happened-before'' relation among the equivalence classes 
based on the local happened-before relation at a single robot 
or a pair of visible robots. 
Our condition also requires the happened-before relation to form 
a directed acyclic graph so that we construct a similar SSYNC execution 
by applying the equivalence classes one by one 
in the order of one of their  topological sort. 

We then present a synchronizer for oblivious ASYNC luminous robots 
to simulate an execution of an algorithm for oblivious 
(non- luminous) SSYNC robots satisfying those conditions, 
as well as some limitations of synchronizers that make use of visible lights.

\noindent{\bf Related work.~}
Existing literature established a rich class of distributed problems 
for mobile robot systems. 
The pattern formation problem~\cite{SY99} is one of the most important 
static problems, that is, 
the robots stop moving once they reach a terminal configuration.  
Suzuki and Yamashita showed 
the oblivious FSYNC robots can solve the rendezvous problem, 
while the oblivious SSYNC robots cannot~\cite{SY99}. 
Flocchini et al. further discussed the rendezvous problem 
to show the power of lights. 
A robot with {\em externally visible light} can change   but cannot see 
the color of its own light, while the other robots can observe it. 
A robot with {\em internally visible light} can change and see the color 
of its own light, while the other robots cannot observe it. 
They showed that the ASYNC robots with externally visible lights can solve 
the rendezvous problem, 
while the SSYNC robots with internally visible lights cannot~\cite{FSVY16}. 
To demonstrate computational power of the luminous robots, 
many dynamic problems has been proposed. 
Das et al. proposed the {\em oscillating points} problem, 
that requires the robots to alternately come closer and go farther 
from each other~\cite{DFPSY16}. 
This problem shows the difference between  luminous ASYNC robots and 
 (non-luminous) FSYNC robots. 
Flocchini et al. examined the power of internally visible lights 
and that of externally visible lights in  FSYNC and SSYNC 
with a variety of static problems such as 
{\em triangle rotation}, 
  {\em center of gravity expansion}, 
and dynamic problems such as 
the {\em perpetual center of gravity expansion} and 
{\em shrinking rotation}~\cite{FSW19}.  

Synchronization was first presented in \cite{DFPSY16} to 
overcome the limit of the ASYNC robots with unlimited visibility. 
In this paper, we further investigate synchronization 
to demonstrate the difference between   ASYNC and SSYNC 
with limited visibility. 

\noindent{\bf Organization.~} 
We provide detailed definitions of ASYNC and SSYNC executions 
and the similarity between two executions in Section~\ref{sec:preliminary}. 
We then provide  a sufficient condition for an ASYNC execution 
to have a similar SSYNC execution, 
and investigate its necessity under a randomized ASYNC  adversary   
in Section~\ref{sec:condition}. 
Section~\ref{sec:synchronizer} provides a synchronizer algorithm 
for oblivious luminous ASYNC robots that satisfy the necessary and sufficient conditions. 
We conclude our paper in Section~\ref{sec:conclusion}.

\section{Preliminary} 
\label{sec:preliminary}

We investigate a system $\mathcal{R}$ of $n$ anonymous oblivious mobile robots 
$\{r_1, r_2, \ldots, r_n \}$ in the 2D space. 
We use $r_i$ just for notation. 
We assume that at most one robot can occupy any given position 
at any time.\footnote{
We can remove this assumption with multiplicity 
detection capability.}

We consider SSYNC and ASYNC 
as the {\em semi-synchronous scheduler} $\SSYNC$ and 
the {\em asynchronous scheduler} $\ASYNC$, respectively. 
We regard a scheduler as a set of schedules that it can produce.  
Consider an infinite execution $E$ of a deterministic algorithm $\mathcal{A}$ 
from an initial configuration $I$ under $\ASYNC$. 
Independently of $\mathcal{A}$ and $I$,\footnote{
Thus, $ASYNC$ produces any schedule including the worst-case 
schedule for $\mathcal{A}$ and $I$.}
$\ASYNC$ nondeterministically produces a schedule $\Omega$, 
which specifies for each $r_i$ when it is activated and 
executes Look-Compute-Move cycles. 

Formally, $\Omega$ is a set of schedules $\Omega_i$ for each robot $r_i$, 
where $\Omega_i$ is an infinite sequence of Look-Compute-Move cycles. 
The $j$th cycle $\omega_i(j)$ of $\Omega_i$ is denoted by a triple 
$(o_i(j), s_i(j), f_i(j))$, 
where $o_i(j)$, $s_i(j)$, and $f_i(j)$ are the time instants 
that $r_i$ takes a snapshot in the Look, 
starts and ends the Move, respectively. 
We assume that the time interval assigned to $\omega_i(j)$ is 
$[o_i(j), f_i(j)]$, 
and $o_i(j) < s_i(j) < f_i(j) < o_i(j+1)$ for all $j=1, 2, \ldots$. 
Scheduler $\Omega$ is {\em fair} in the sense that each $\Omega_i$ 
satisfies that, for any $t \in \mathbf{R}^+$, there is a $j \in \mathbf{N}$ 
such that $o_i(j) > t$, where 
$\mathbf{R}^+$ and $\mathbf{N}$ are the sets of positive real numbers 
and non-negative integers, respectively.
The execution is not uniquely determined 
by $I$, $\mathcal{A}$, and $\Omega$, 
due to the only source of non-determinism, that is, 
non-rigid movement of robots. 
The set of possible executions of $\mathcal{R}$ given $I$, $\mathcal{A}$, and $\Omega$ 
is denoted by $\mathcal{E}(\Omega, \mathcal{A}, I)$. 

The visibility range of each robot is the unit distance 
of the global coordinate system $\mathcal{Z}_0$.\footnote{The common 
visibility range does not promise common unit distance among the robots.} 
The snapshot $P_i(j)$ taken by $r_i$ at $o_i(j)$ in $\omega_i(j)$ 
is the set of positions of robots in $\mathcal{Z}_i$ visible from $r_i$ at $o_i(j)$. 
$P_i(j)$ always contains its origin, 
because it is the position of $r_i$ in $\mathcal{Z}_i$. 
The number of robots visible from $r_i$ at $o_i(i)$ is denoted by 
$|P_i(j)|$. 
Let $P_i(E) = \{P_i(j) \mid j \in {\mathbf N} \}$ and 
$P(E) = \{P_i(E) \mid r_i \in \mathcal{R}\}$. 

The position of $r_i$ at $o_i(j)$ in the global coordinate system $\mathcal{Z}_0$ 
is denoted by $\pi_i(j)$. 
Note that $r_i$ cannot recognize its position $\pi_i(j)$.
The {\em footprint} of $r_i$ in $E$ is 
$\Pi_i(E) = \{\pi_i(j) \mid j \in {\mathbf N} \}$ 
and let $\Pi(E) = \{\Pi_i(E) \mid r_i \in \mathcal{R}\}$ be the set of footprints 
of all robots of $\mathcal{R}$. 
Since the system is not rigid, $\pi_i(j+1)$ is not uniquely determined by 
$\pi_i(j)$, $P_i(j)$, $\mathcal{Z}_i$, $\mathcal{A}$, and $I$. 

Let $E \in \mathcal{E}(\Omega, \mathcal{A}, I)$ and 
$\tilde E \in \mathcal{E}(\tilde \Omega, \mathcal{A}, I)$. 
If the system is rigid, 
$\Pi(E) = \Pi(\tilde E)$ if $P(E) = P(\tilde E)$. 
Otherwise, for some $\Omega$, $\mathcal{A}$, and $I$, 
there are executions $E, \tilde E \in \mathcal{E}(\Omega, \mathcal{A}, I)$ 
such that $P(E) = P(\tilde E)$ and $\Pi(E) \neq \Pi(\tilde E)$. 

\begin{ex}
When $n=1$, no matter where $r_1$ goes, $P_1(j) = \{(0,0)\}$. 
When $n=2$, suppose that $r_1$ and $r_2$ are initially at $(0,0)$ and $(0,1)$, 
respectively, and synchronously move in parallel along the $x$-axis of $\mathcal{Z}_0$ 
at the same speed. 
As long as their tracks are truncated at the same $x$-coordinate, 
independently of where they are truncated, 
$P_i(j) = P_i(j')$ for all $i \in \{1,2\}$ and $i, j' \in \mathbf{N}$. 
\end{ex} 

We say two executions $E \in \mathcal{E}(\Omega, \mathcal{A}, I)$ and 
$\tilde E \in \mathcal{E}(\Omega, \mathcal{A}, I)$ are {\em similar}, 
denoted by $E \sim \tilde E$, if $P(E) = P(\tilde E)$ and 
$\Pi(E) = \Pi(\tilde E)$. 

Without loss of generality, we assume that 
$\SSYNC$ produces a schedule $\Omega$ 
such that every cycle $\omega_i(j)$ has a form $(t, t+1/4, t+3/4)$ 
for some $t \in \mathbf{N}$.

\section{ASYNC execution with a similar SSYNC execution}
\label{sec:condition}

We provide a necessary and sufficient condition for 
an ASYNC execution $E \in \mathcal{E}(\Omega, \mathcal{A}, I)$ to have 
an SSYNC execution $\tilde E \in \mathcal{E}(\tilde \Omega, \mathcal{A}, I)$ 
such that $P(E) = P(\tilde E)$ and $\Pi(E) = \Pi (\tilde E)$ 
for some $\tilde \Omega \in \SSYNC$.

\subsection{Sufficiency} 

In an SSYNC execution, at each discrete time $t=0, 1, 2, \ldots$ 
the activated robots execute a Look-Compute-Move cycle 
with each of the three phases completely synchronized. 
From this definition, we directly obtain the following three 
assumptions on an ASYNC execution to have a 
corresponding SSYNC execution. 
(i) No robot observes other robots moving. 
(ii) When two robots observe each other, 
they are activated at the same time in a SSYNC execution. 
(iii) Due to limited visibility, 
the above ``mutually observed'' relationship is transitive. 
We formally describe these assumptions. 

Let $S_i(j)$ be the set of robots visible from $r_i$ at time $o_i(j)$. 
Then, $|S_i(j)| = |P_i(j)|$ and $r_i \in S_i(j)$. 
Recall that $r_i$ cannot recognize the correspondence between 
$S_i(j)$ and $P_i(j)$. 
We say that $E$ is stationary, 
if every snapshot $P_i(j)$ in $E$ is 
``stationary'' in the sense that $r_i$ does not observe another 
robot $r_{i'}$ in its move phase. 
Formally, $E \in \mathcal{E}(\Omega, \mathcal{A}, I)$ 
is {\em stationary} if $o_i(j) \not\in (s_{i'}(j'), f_{i'}(j'))$ 
holds for any pair of cycles $\omega_i(j)$ and $\omega_{i'}(j') \in \Omega$ 
such that $r_{i'} \in S_i(j)$.\footnote{We exclude $s_{i'}(j')$ 
and $f_{i'}(j')$, because $r_{i'}$ is not moving at these time instants.} 

\begin{assumption}
\label{ass1}
We assume that $E$ is stationary.
\end{assumption}

Let $\omega_i(j)$ and $\omega_{i'}(j') \in \Omega$ be two cycles such that
$i \neq i'$ and $o_i(j) \leq o_{i'}(j')$. 
If $[o_i(j), f_i(j)] \cap [o_{i'}(j'), f_{i'}(j')] \neq \emptyset$ 
and $r_i \in S_{i'}(j')$, 
we say that $\omega_i(j)$ and $\omega_{i'}(j')$ {\em overlap each other}. 

We say $\omega_i(j)$ and $\omega_{i'}(j')$ are {\em concurrent}, 
denoted by $\omega_i(j) \parallel \omega_{i'}(j')$, 
if one of the following conditions holds: 
\begin{enumerate}
\item 
$i = i'$ and $j = j'$.  
\item 
$i \neq i'$, $o_i(j) \in (f_{i'}(j'-1), o_{i'}(j')]$, 
$o_{i'}(j') \in [o_i(j), s_i(j)]$, and $r_{i'} \in S_i(j)$ 
(thus, $r_i \in S_{i'}(j')$). 
\item 
$i \neq i'$, $o_{i'}(j') \in (f_i(j-1), o_i(j)]$, 
$o_i(j) \in [o_{i'}(j'), s_{i'}(j')]$, and $r_i \in S_{i'}(j')$ 
(thus, $r_{i'} \in S_i(j)$). 
\end{enumerate}
The concurrency relation $\parallel$ is symmetric and reflexive, 
but is not always transitive. 
By definition, we have the following proposition. 

\begin{proposition}
\label{prop:1}
For any $i$, $j$, and $j'$, 
$\omega_i(j) \parallel \omega_i(j')$ if and only if $j=j'$. 
\end{proposition}

If $i \neq i'$ and $\omega_i(j)$ and $\omega_{i'}(j')$ are concurrent, 
then they overlap each other. 
Moreover, both $r_i$ and $r_j$ observe each other in $\omega_i(j)$ and 
$\omega_{i'}(j')$, respectively. 
If $\omega_{i}(j)$ and $\omega_{i'}(j')$ overlap each other, 
the two robots do not always observe each other. 
However, at least one of them observes the other. 

We say that $E$ is {\em pairwise aligned}, 
if two cycles $\omega_i(j)$ and $\omega_{i'}(j')$ that overlap 
each other are always concurrent. 

\begin{assumption}
\label{ass2}
We assume that $E$ is pairwise aligned. 
\end{assumption}

Let $\stackrel{*}{\parallel}$ be the transitive closure of $\parallel$, 
that is an equivalence relation on $\Omega$. 
We abuse the term so that $\omega_i(j)$ and $\omega_{i'}(j')$ are 
{\em concurrent} if $\omega_i(j) \stackrel{*}{\parallel} \omega_{i'}(j')$. 
Since $\parallel$ is not always transitive, 
$\parallel \neq \stackrel{*}{\parallel}$ may hold.
Let $\Omega = \Omega_{0}, \Omega_{1}, \Omega_{2}, \ldots$ 
be the equivalence class partition of $\Omega$ with respect to 
$\stackrel{*}{\parallel}$. 
Intuitively, the cycles in each $\Omega_i$ must be executed 
at the same time in the corresponding SSYNC execution. 
However, because $\stackrel{*}{\parallel}$ is transitive, 
we need to consider observations, i.e., 
$S_i(j)$ and $S_{i'}(j')$ for each $\omega_i(j), \omega_{i'}(j') \in \Omega_i$. 

Let $dist(p,q)$ denote the Euclidean distance between two points 
$p$ and $q$ in $\mathcal{Z}_0$. 
We say that $E$ is {\em consistent}, 
if the following conditions hold for any pair of cycles $\omega_i(j)$ 
and $\omega_{i'}(j')$ such that 
$\omega_i(j) \stackrel{*}{\parallel} \omega_{i'}(j')$: 
\begin{enumerate}
\item 
$r_{i'} \in S_i(j)$ if and only if $r_i \in S_{i'}(j')$.
\item
If $r_{i'} \in S_i(j)$, or equivalently $r_i \in S_{i'}(j')$, 
$\omega_i(j) \parallel \omega_{i'}(j')$.
\item
If $r_{i'} \not\in S_i(j)$, or equivalently $r_i \not\in S_{i'}(j')$,
$dist(\pi_i(j), \pi_{i'}(j')) > 1$.
\end{enumerate}


\begin{assumption}
\label{ass3}
We assume that $E$ is consistent. 
\end{assumption}

The following proposition is an extension of Proposition~\ref{prop:1}. 
\begin{proposition}
\label{prop:2}
Suppose that $E$ is stationary, pairwise aligned, 
and consistent. 
For any $i$, $j$, and $j'$, 
$\omega_{i}(j) \stackrel{*}{\parallel} \omega_{i}(j')$ if and only if $j=j'$. 
\end{proposition}
\begin{proof}
If $j = j'$, $\omega_i(j) \stackrel{*}{\parallel} \omega_i(j')$ for any $i$. 
Otherwise, suppose that $\omega_i(j) \stackrel{*}{\parallel} \omega_i(j')$ 
holds. 
Since $r_i \in S_i(j')$ for any $i$ and $j'$, 
$\omega_i(j) \parallel \omega_i(j')$ by the consistency, 
which is a contradiction by Proposition~\ref{prop:1}. 
Thus, if $j \neq j'$, then 
$\omega_i(j) \not\stackrel{*}{\parallel} \omega_i(j')$. 
\qed 
\end{proof}

Let $\omega_i(j), \omega_{i'}(j') \in \Omega$ be two cycles. 
We say that $\omega_i(j)$ happens {\em immediately before} $\omega_{i'}(j')$, 
denoted by $\omega_i(j) \rightarrow \omega_{i'}(j')$, 
if one of the following conditions holds: 
\begin{enumerate}
\item 
$i' = i$ and $j' = j+1$. 
\item 
$i' \neq i$, $r_i \in S_{i'}(j')$ and $o_{i'}(j') \in (f_i(j), s_i(j+1)]$. 
\item 
$i' \neq i$, $r_{i'} \in S_i(j)$ and 
$f_{i'}(j'-1) < o_i(j) < f_i(j) < o_{i'}(j')$. 
\end{enumerate}

We may call $\rightarrow$ a ``happened-before'' relation on $\Omega$, 
because $\omega_i(j) \rightarrow \omega_{i'}(j')$ denotes the fact 
that $\omega_i(j)$ happens before $\omega_{i'}(j')$. 
Thus, $\rightarrow$ is neither reflexive nor symmetric. 
It is worth emphasizing that in general, $\rightarrow$ is not transitive 
either, because it is defined based on the visibility relation between robots 
like $\parallel$. 

Note that if $\omega_i(j) \rightarrow \omega_{i'}(j')$, 
then $\omega_i(j) \not\parallel \omega_{i'}(j')$, 
because they do not overlap each other. 
Note also that either $r_i \in S_{i'}(j')$ or $r_{i'} \in S_i(j)$ holds, 
but generally not both. 

\begin{proposition}
\label{prop:3} 
Suppose that $E$ is stationary, pairwise aligned, 
and consistent. 
No equivalence class $\Omega_k$ of $\Omega$ contains cycles 
$\omega_i(j)$ and $\omega_{i'}(j')$ such that 
$\omega_i(j) \rightarrow \omega_{i'}(j')$. 
\end{proposition}
\begin{proof}
Let $\omega_i(j)$ and $\omega_{i'}(j')$ be any cycles in $\Omega_k$, 
i.e., $\omega_i(j) \stackrel{*}{\parallel} \omega_{i'}(j')$. 
We assume that $\omega_i(j) \rightarrow \omega_{i'}(j')$. 

If $i' = i$, then $j' = j$ by Proposition~\ref{prop:2}, 
which is a contradiction. 
Thus, $i' \neq i$. 
Since $\omega_i(j) \rightarrow \omega_{i'}(j')$, 
either $r_i \in S_{i'}(j')$ or $r_{i'} \in S_i(j)$ holds. 
This implies $\omega_i(j) \parallel \omega_{i'}(j')$ by the consistency. 
This is a contradiction. 
\qed
\end{proof}

Let $\Omega_k$ and $\Omega_{k'}$ be two equivalence classes of $\Omega$. 
If there are cycles $\omega_i(j) \in \Omega_k$ and 
$\omega_{i'}(j') \in \Omega_{k'}$ 
such that $\omega_i(j) \rightarrow \omega_{i'}(j')$, 
we use the notation $\Omega_k \Rightarrow \Omega_{k'}$. 
By Proposition~\ref{prop:3}, binary relation $\Rightarrow$ is not reflexive. 
We say that $E$ is {\em serializable} if the 
infinite graph 
$\mathcal{G} = (\{\Omega_0, \Omega_1, \Omega_2, \ldots\}, \Rightarrow)$ is acyclic. 
\begin{assumption}
\label{ass4}
We assume that $E$ is serializable. 
\end{assumption}

While $\rightarrow$ can be considered as a ``happened-before'' relation 
on cycles, 
it is not adequate to consider $\Rightarrow$ as a ``happened-before'' relation 
on the equivalence classes. 
There can be cycles $\omega_i(j) \in \Omega_k$ and $\omega_{i'}(j') \in \Omega_{k'}$ 
such that $f_{i'}(j') < o_i(j)$ even if $\Omega_k \Rightarrow \Omega_{k'}$. 
In fact, $\Omega_k \Rightarrow \Omega_{k'}$ and 
$\Omega_k \Leftarrow \Omega_{k'}$ may hold at the same time. 
There is a stationary, pairwise aligned, 
and consistent execution  
which is not serializable. 

The following proposition is clear by definition. 
\begin{proposition}
\label{prop:4}
Suppose that $E$ is stationary, pairwise aligned, 
consistent, and 
serializable. 
If $\omega_i(j)$ and $\omega_i(j')$ are in an equivalence class $\Omega_k$ 
for some $i$, $j$, and $j'$, then $j=j'$. 
That is, each $\Omega_k$ contains at most one cycle for every robot $r_i$. 
\end{proposition}

Let $E \in \mathcal{E}(\Omega, \mathcal{A}, I)$ be an execution of algorithm $\mathcal{A}$ from 
initial configuration $I$ under a schedule $\Omega \in \ASYNC$ and 
assume that $E$ is stationary, pairwise aligned, 
consistent, and serializable. 
Let $\mathcal{G} = (\{\Omega_0, \Omega_1, \Omega_2, \ldots\}, \Rightarrow)$ 
be the acyclic graph obtained from $E$. 
Without loss of generality, 
let 
$T = (\Omega_0, \Omega_1, \Omega_2, \ldots)$ be a topological sort of $\mathcal{G}$. 
Let $A_k = \{r_i \mid \omega_i(j) \in \Omega_k \ \text{for some $j$}\}$, 
keeping in mind that $\omega_i(j)$ is unique for each $r_i \in A_k$ 
by Proposition~\ref{prop:4}. 

We construct a schedule $\tilde{\Omega} \in \SSYNC$ from $T$. 
Intuitively, $\tilde{\Omega}$ activates all robots in $A_k$ at time $k$ for all 
$k \in \mathbf{N}$. 
Formally, $\tilde{\Omega} = \{\tilde{\omega}_i(j) \mid 
i=1,2,\ldots,n, \ j \in \mathbf{N} \}$, where 
$\tilde{\omega}_i(j) = (\tilde{o}_i(j), \tilde{s}_i(j), \tilde{f}_i(j)) = (k, k+1/4, k+3/4)$ 
if $\omega_i(j) \in \Omega_k$ for $k \in \mathbf{N}$. 
Obviously, $\tilde{\Omega}$ is well-defined and $\tilde{\Omega} \in \SSYNC$. 

Consider $\tilde{E} \in \mathcal{E}(\tilde{\Omega}, \mathcal{A}, I)$. 
Let $\tilde{S}_i(j)$ and $\tilde{P}_i(j)$ be the set of robots 
visible from $r_i$ at $\tilde{o}_i(j)$ and the snapshot that $r_i$ takes 
at $\tilde{o}_i(j)$, respectively. 
Then, $P(\tilde{E}) = \{\tilde{P}_i(j) \mid i = 1,2, \ldots, n, \ \text{and} \ j \in \mathbf{N}\}$. 
Let $\tilde{\pi}_i(j)$ be the position of $r_i$ in $\mathcal{Z}_0$ at $\tilde{o}_i(j)$. 
Then, $\Pi(\tilde{E}) = \{\tilde{\pi}_i(j) \mid i = 1, 2, \ldots, n, \ \text{and} \ j \in {\mathbf{N}}\}$. 

Finally, we need to examine each pair of cycles $\omega_i(j)$ and $\omega_{i'}(j')$ 
assigned to different discrete time in $\tilde{\Omega}$. 
Thus, $\omega_i(j) \not\stackrel{*}{\parallel} \omega_{i'}(j')$. 
Consider the case where 
$\omega_i(j) \in \Omega_k$ and $\omega_{i'}(j') \in \Omega_{k''}$ for 
$k < k''$ and $r_{i'}$ executes no cycle 
during $\Omega_{k}, \Omega_{k+1}, \ldots, \Omega_{k''-1}$. 
Thus, $r_{i'}$ does not move during the time period $[k, k''-1]$ 
in $\tilde{\Omega}$. 
If $r_i$ observes $r_{i'}$ in $\omega_i(j)$, 
$\omega_i(j)$ and $\omega_{i'}(j')$ do not overlap each other. 
Otherwise, $\Omega_k$ contains $\omega_{i'}(j')$. 
If $r_i$ does not observe $r_{i'}$ in $\omega_i(j)$, 
$dist(\pi_i(j), \pi_{i'}(j')) > 1$. 
Otherwise, $r_{i'}$ need to move during the time period $[k, k''-1]$ 
in $\tilde{\Omega}$. 
We describe this situation with the notion of naturality. 

For any $\omega_i(j)$ and $i'(\neq i)$, 
there is a $j'$ such that $k' \leq k < k''$, 
where $\omega_i(j) \in \Omega_k$, $\omega_{i'}(j'-1) \in \Omega_k'$, 
and $\omega_{i'}(j') \in \Omega_{k''}$.\footnote{We consider $k' = -1$ 
if $j' = 1$ and $\omega_{i'}(j')$ is not defined.} 
A topological sort $T = (\Omega_0, \Omega_1, \Omega_2, \ldots)$ is 
{\em natural} if the following conditions hold for any pair of such cycles 
$\omega_i(j)$ and $\omega_{i'}(j')$: 
\begin{enumerate}
\item 
Suppose that $r_{i'} \in S_i(j)$.  
Thus $\omega_i(j) $ and $\omega_{i'}(j')$ do not overlap each other.
Then $o_i(j) < o_{i'}(j')$.
\item 
Suppose that $r_{i'} \not\in S_i(j)$. 
Then $dist(\pi_i(j), \pi_{i'}(j')) > 1$. 
\end{enumerate}
We say that $E$ is {\em natural} if $E$ has a natural topological sort. 

\begin{assumption}
\label{ass5}
We assume that $E$ is natural. 
\end{assumption}

\begin{theorem}
\label{theorem:sufficiency} 
If an execution $E \in \mathcal{E}(\Omega, \mathcal{A}, I)$ satisfies 
Assumptions~\ref{ass1}, \ref{ass2}, \ref{ass3}, \ref{ass4}, \ref{ass5}, 
there is an execution $\tilde{E} \in \mathcal{E}(\tilde{\Omega}, \mathcal{A}, I)$ 
for some $\tilde{\Omega} \in \SSYNC$ such that $E \sim \tilde{E}$, 
i.e., $P(E) = P(\tilde{E})$ and $\Pi(E) = \Pi(\tilde{E})$. 
\end{theorem}
\begin{proof}
We construct an execution $\tilde{E} \in \mathcal{E}(\tilde{\Omega}, \mathcal{A}, I)$ 
as follows: 
If $r_i$ moves in $\omega_i(j)$ from $\pi_i(j)$ to $\pi_i(j+1)$ in $E$, 
we move $r_i$ in $\tilde{\omega}_i(j)$ from $\pi_i(j)$ to $\pi_i(j+1)$ 
to construct $\tilde{E}$ that satisfies the conditions. 
We guarantee its feasibility by showing that $P_i(j) = \tilde{P}_i(j)$ and 
$\pi_i(j) = \tilde{\pi}_i(j)$. 
Indeed, if $P_i(j) = \tilde{P}_i(j)$, 
$\pi_i(j) = \tilde{\pi}_i(j)$, and $r_i$ can move from $\pi_i(j)$ 
to $\pi_i(j+1)$ during $\omega_i(j)$ in $E$, 
then it can move to the same position in $\tilde{\omega}_i(j)$. 

Suppose that a topological sort $T = (\Omega_0, \Omega_1, \Omega_2, \ldots)$ 
is natural. 
It suffices to show the following claim for all $k=0, 1, 2, \ldots$: 
For any robot $r_i \in A_k$, 
$S_i(j) = \tilde{S}_i(j)$, $P_i(j) = \tilde{P}_i(j)$, 
$\pi_i(j) = \tilde{\pi}_i(j)$, and $\pi_i(j+1) = \tilde{\pi}_i(j+1)$ hold. 
The proof is by induction on $k$. 

\noindent{\bf Base case ($k=0$).~} 
Let $r_i \in A_0$ be any robot. 
Then, $\omega_i(j) \in \Omega_0$ holds for some $j$ and obviously $j=1$. 
We show that $S_i(1) = \tilde{S}_i(1)$. 
Recall that $\tilde{o}_i(j) = 0$. 

Let $r_{i'}$ be any robot. Suppose that $r_{i'} \in A_0$, i.e., 
$\omega_{i'}(1) \in \Omega_0$ and 
$\omega_i(1) \stackrel{*}{\parallel} \omega_{i'}(1)$. 
If $r_{i'} \in S_i(1)$, then $\omega_i(1) \parallel \omega_{i'}(1)$ 
by the consistency, which implies $r_{i'} \in \tilde{S}_i(1)$. 
Otherwise, i.e., $r_{i'} \not\in S_i(1)$, $r_i \not\in S_{i'}(1)$, and thus 
$dist(\pi_i(1), \pi_{i'}(1)) > 1$ by the consistency. 
Thus, $r_{i'} \not\in \tilde{S}_i(1)$. 

Next, suppose that $r_{i'} \not\in A_0$. 
Then, $\omega_{i'}(1) \in \Omega_{k'}$ for some $k' > 0$. 
If $r_{i'} \in S_i(1)$, then $o_i(1) < o_{i'}(1)$ by the naturality. 
Since $r_i$ and $r_{i'}$ do not move until $o_i(1)$, 
$r_{i'} \in \tilde{S}_i(1)$. 
Otherwise, i.e., $r_{i'} \not\in S_i(1)$, $dist(\pi_i(1), \pi_{i'}(1)) > 1$ 
by the naturality. Thus, $r_{i'} \not\in \tilde{S}_i(1)$. 

Since $E$ and $\tilde{E}$ start with the same initial configuration $I$, 
for any $r_i \in A_0$, $\pi_i(1) = \tilde{\pi}_i(1)$, and hence 
$P_i(1) = \tilde{P}_i(1)$ since $S_i(1) = \tilde{S}_i(1)$. 
Since $\pi_i(2)$ is reachable from $\pi_i(1)$ by algorithm $\mathcal{A}$
when the snapshot is $P_i(1)$, 
$\tilde{\pi}_i(2)(= \pi_i(2))$ is reachable from 
$\tilde{\pi}_i(1)(= \pi_i(1))$ by algorithm $\mathcal{A}$ 
when the snap shot is $\tilde{P}_i(1)(= P_i(1))$. 
Thus, to construct $\tilde{E}$, we move $r_i$ to 
$\tilde{\pi}_i(2)(= \pi_i(2))$ from $\tilde{\pi}_i(1)(= \pi_i(1))$ 
in $\tilde{\omega_i}(1)$. 

\noindent{\bf Induction step.~}
Assume that the claim holds for all $0 \leq k < K$, 
and we show that the claim folds for $k=K$. 
Let $r_i \in A_k$ be any robot such that $\omega_i(j) \in \Omega_K$. 

If $\omega_i(j-1) \in \Omega_k$, then $k < K$ and $\pi_i(j) = \tilde{\pi}_i(j)$ 
by the induction hypothesis. 
Indeed, for all $i'$ and $j'$, if $\omega_{i'}(j') \in \Omega_{k}$ for some 
$k \leq K$, $\pi_{i'}(j') = \tilde{\pi}_{i'}(j')$ by the induction hypothesis. 

Consider any robot $r_{i'} \in S_i(j)$. 
Since $E$ is stationary, $r_{i'}$ is not in move phase at $o_i(j)$. 
Suppose $o_i(j) \in (f_{i'}(j'), o_{i'}(j'+1))$ for some $j'$. 
Then, $\omega_{i'}(j') \rightarrow \omega_i(j)$ and 
$\Omega_{K'} \Rightarrow \Omega_K$, 
where $\omega_{i'}(j') \in \Omega_{K'}$ for some $K' < K$. 
Since (i) $\pi_i(j) = \tilde{\pi}_i(j)$, 
(ii) $\pi_{i'}(j'+1) = \tilde{\pi}_{i'}(j'+1)$ because $K' < K$, and 
(iii) the position of $r_{i'}$ (in $\mathcal{Z}_0$) at $o_i(j)$ is $\pi_{i'}(j'+1)$, 
we have $r_{i'} \in \tilde{S}_i(j)$ and $p_{i'} = \tilde{p}_{i'}$, 
where $p_{i'}$ and $\tilde{p}_{i'}$ are positions of $r_{i'}$ (in $\mathcal{Z}_i$) 
in $P_i(j)$ and $\tilde{P_i(j)}$, respectively. 
Here we use the fact that either $\omega_i(j) \parallel \omega_{i'}(j'+1)$ 
or $\omega_i(j) \rightarrow \omega_{i'}(j'+1)$ holds, 
and $K \leq K''$ holds, where $\omega_{i'}(j'+1) \in \Omega_{K''}$. 

Suppose otherwise $o_i(j) \in [o_{i'}(j'), s_{i'}(j')]$ for some $j'$. 
Then $\omega_i(j) \parallel \omega_{i'}(j')$ and 
$\omega_{i'}(j') \in \Omega_K$. 
Let $\omega_i(j-1) \in \Omega_k$ and $\omega_{i'}(j'-1) \in \Omega_{k'}$. 
Then, $k, k' < K$. 
Thus, $\pi_{i'}(j') = \tilde{\pi}_i(j)$ by the induction hypothesis. 
Since $\pi_{i'}(j')$ is the position of $r_{i'}$ at $o_i(j)$ and 
$\pi_i(j) = \tilde{\pi}_i(j)$, $r_{i'} \in \tilde{S}_{i'}(j')$, 
and $p_{i'} = \tilde{p}_{i'}$. 

Finally, consider any robot $r_{i'} \in \tilde{S}_i(j)$ to confirm 
$S_i(j) = \tilde{S}_i(j)$. 
If there is a $j'$ such that $\omega_i(j) \stackrel{*}{\parallel} \omega_{i'}(j')$, 
$dist(\pi_i(j), \pi_{i'}(j')) > 1$ by the consistency. 
This implies that $r_{i'} \not\in \tilde{S}_i(j)$ by the induction hypothesis. 

Otherwise, $\omega_i(j) \not\stackrel{*}{\parallel} \omega_{i'}(\ell)$ 
for all $\ell \in \mathbf{N}$. 
Let $j'$ be an integer such that $K' < K < K''$, 
where $\omega_i(j) \in \Omega_K$, $\omega_{i'}(j'-1) \in \Omega_{K'}$, 
and $\omega_{i'}(j') \in \Omega_{K''}$. 
By the naturality, $dist(\pi_i(j), \pi_{i'}(j')) > 1$, 
which implies that $r_{i'} \not\in S_i(j)$ by the induction hypothesis. 

Since $\pi_i(j) = \tilde{\pi}_i(j)$, the position of $r_{i'}$ in $\tilde{P}_i(j)$ 
and $\tilde{P}_i(j)$ (both in $\mathcal{Z}_i$) are the same, i.e., 
$P_i(j) = \tilde{P}_i(j)$. 
By the same reason as the base case, 
to construct $\tilde{E}$, we move $r_i$ to $\tilde{\pi}(j+1) (=\pi_i(j+1))$ 
from $\tilde{\pi}_i(j) (= \pi_i(j))$ in $\tilde{\omega}_i(j)$. 
\qed 
\end{proof}

\subsection{Necessity} 

The conjunction of conditions in 
Assumptions~\ref{ass1}, \ref{ass2}, \ref{ass3}, \ref{ass4}, and \ref{ass5} 
is a sufficient condition for an ASYNC execution 
$E \in \mathcal{E}(\Omega, \mathcal{A}, I)$ to have a similar SSYNC execution 
for some $\tilde{\Omega} \in \SSYNC$. 
However, it is not necessary in general. 
Suppose that $\mathcal{A}$ does not move any robot. 
Then, the robot system is at its initial configuration $I$ forever, 
and every execution $E$ has a similar SSYNC execution $\tilde{E}$, 
regardless of whether or not $E$ satisfies each of 
the five assumptions. 
We will show the necessity of the five assumptions 
assuming a randomized adversary that determines rigid movement and 
when moving robots are observed. 

Let $\tau_i(j)$ in $\mathcal{Z}_0$ be the route of $r_i$ computed by $\mathcal{A}$ 
in $\omega_i(j)$ given snapshot $P_i(j)$ in $\mathcal{Z}_i$ as input. 
We assume that $\tau_i(j)$ is a simple curve such that $|\tau_i(j)| > \delta$,
where a curve is said to be {\em simple} if it does not contain an intersection.
Recall that $r_i$ never stops en route if $|\tau_i(j)| \leq \delta$. 
Otherwise, $r_i$ travels an arbitrary initial part $\hat{\tau}_i(j)$ of $\tau_i(j)$ 
such that $|\hat{\tau}_i(j)| > \delta$ at an arbitrary (possibly variable) speed 
during Move in $\omega_i(j)$. 
Thus, another robot $r_{i'}$ can observe $r_i$ at any position $y$ 
in $\hat{\tau}_i(j)$ (thus, $\tau_i(j)$) 
at $o_{i'}(j')$ if $r_i \in S_{i'}(j')$ and 
$o_{i'}(j') \in (s_i(j), f_i(j))$. 

Intuitively, we assume that $\hat{\tau}_i(j)$ satisfying $|\hat{\tau}_i(j)| > \delta$ 
is chosen ``uniformly at random'' and 
that $y$ is chosen ``uniformly at random'' from $\tau_i(j)$ 
(provided that $\hat{\tau}_i(j) = \tau_i(j)$). 
Then, we show that if an execution $E \in \mathcal{E}(\Omega, \mathcal{A}, I)$ has 
a similar execution $\tilde{E} \in \mathcal{E}(\tilde{\Omega}, \mathcal{A}, I)$ for some 
$\tilde{\Omega} \in \SSYNC$, 
then $E$ satisfies each of the five conditions in 
Assumptions~\ref{ass1}, \ref{ass2}, \ref{ass3}, \ref{ass4}, and \ref{ass5} 
with ``probability'' $1$.

\subsubsection{Stationarity} 

For a fixed algorithm $\mathcal{A}$ and initial configuration $I$, 
we show that the stationarity is necessary. 
To make the argument simple, 
we assume that the system is rigid. 
Then, $E \in \mathcal{E}(\Omega, \mathcal{A}, I)$ is uniquely determined by 
$\Omega \in \ASYNC$ if $E$ is stationary. 

Consider any schedule $\Omega \in \ASYNC$ that contains a unique pair of cycles 
$\omega_i(j)$ and $\omega_{i'}(j')$ such that $o_i(j) \in (s_{i'}(j'), f_{i'}(j'))$. 
Let $\tau_{i'}(j')$ be the route (in $\mathcal{Z}_0$) that $\mathcal{A}$ at $r_{i'}$ computes 
given $P_{i'}(j')$ (in $\mathcal{Z}_{i'}$). 
Then, $E \in \mathcal{E}(\Omega, \mathcal{A}, I)$ is uniquely determined by the position 
$y \in \tau_{i'}(j')$ of $r_{i'}$ at $o_i(j)$ 
(regardless of whether or not $y$ is visible from $r_i$ at $o_i(j)$). 
We normalize $y$ by $z = |\hat{\tau}_{i'}(j')|/|\tau_{i'}(j')| \in [0,1]$, 
where $\hat{\tau}_{i'}(j')$ is the prefix of $\tau_{i'}(j')$ before (and including) $y$. 
Then, there is a one-to-one correspondence between $y \in \tau_{i'}(j')$ and $z \in [0,1]$, 
because $\tau_{i'}(j')$ is a simple curve. 
We use the notation $E(z)$ to emphasize that $E$ is determined by $z$. 

We use a Borel measurable space $([0,1], \mathcal{B}([0,1]))$, 
where $\mathcal{B}([0,1])$ is the Borel $\sigma$-algebra on $[0,1]$, 
i.e., the smallest 
$\sigma$-algebra containing all open intervals in $[0,1]$.
Then, the probability measure $\lambda(\cdot)$ 
for each element of $\mathcal{B}([0,1])$
is the (1-dimensional) Lebesgue probability measure on $[0,1]$,
that satisfies,
for all intervals $(a,b)$ where $0 \leq a < b \leq 1$, $\lambda({(a,b)}) = b-a$.
Thus, $\{[a,a]\}$ and $\{ [a,a], [b,b]\}$ are $\lambda$-null sets. 
In the following, we consider a probability space 
$([0,1], \mathcal{B}([0,1]), \lambda)$. 


Let $\mathcal{Y} = \{y \in \tau_{i'}(j') \mid dist(y, \pi_i(j)) \leq 1\}$ 
be the set of positions $y$ in $\tau_{i'}(j')$ visible from $r_i$ at $o_i(j)$, 
and let $\mathcal{D}$ be the set of $z'$s corresponding to the $y$'s in $\mathcal{Y}$. 
Then, $E(z)$ is not stationary if and only if $z \in \mathcal{D}$. 
Since $\tau_{i'}(j')$ is continuous, $\mathcal{D} \in \mathcal{B}([0,1])$. 
We assume that the probability that $z \in \mathcal{D}$ is $\lambda(\mathcal{D})$. 
Thus $\lambda(\mathcal{D}) = 0$ means that $r_{i'}$ is not visible from $r_i$ at $o_i(j)$
and hence the stationarity is not violated with probability 1
(since we assume that the value of $z$ is chosen uniformly at random from $[0,1]$).
Here and in the rest of this section, 
whenever we say that an execution violates a condition,
we assume that the condition is violated with positive probability.
That is, we assume here that $\lambda(\mathcal{D}) > 0$;
the violation occurs with a positive probability.

\begin{lemma}
\label{lemma1} 
Suppose that $\lambda(\mathcal{D}) > 0$. 
Then, $E(z) \not\sim \tilde{E} \in \mathcal{E}(\tilde{\Omega}, \mathcal{A}, I)$ 
for any $\tilde{\Omega} \in \SSYNC$ $\lambda$-almost everywhere on $\mathcal{D}$, 
i.e., $E(z)$ does not have a similar SSYNC execution $\tilde{E}$ 
for all $z \in \mathcal{D}$, except for a countable number of exceptions.  
\end{lemma} 
\begin{proof}
Clearly, $E(z) \not\sim E(z')$ if $z \neq z'$ provided $z, z' \in \mathcal{D}$. 
Thus, $\{E(z) \mid z \in \mathcal{D}\}$ is uncountable 
since $\lambda(\mathcal{D}) > 0$. 

On the other hand, $\SSYNC$ is a countable set. 
Since the system is rigid, $\tilde{E} \in \mathcal{E}(\tilde{\Omega}, \mathcal{A}, I)$ 
is uniquely determined by $\tilde{\Omega} \in \SSYNC$ 
(because $\mathcal{A}$ and $I$ are fixed), 
$\{\tilde{E} \in \mathcal{E}(\tilde{\Omega}, \mathcal{A}, I) \mid \tilde{\Omega} \in \SSYNC\}$ 
is countable. 
Thus, only for a countable number of executions $E(z)$, 
there are similar SSYNC executions $\tilde{E}$.
\qed 
\end{proof}

The above lemma states the following: 
If the probability that $r_{i'}$ is visible from $r_i$ at $o_i(j)$ is not $0$, 
then the probability that $E$ has a similar SSYNC execution $\tilde{E}$ is $0$, 
under the condition that $r_{i'}$ is indeed visible from $r_i$ at $o_i(j)$
and the stationarity is indeed violated. 

We extend Lemma~\ref{lemma1} to the case in which $\Omega$ contains 
more than one pair of cycles $\omega_i(j)$ and $\omega_{i'}(j')$ 
such that $o_i(j) \in (s_{i'}(j'), f_{i'}(j'))$. 
We order the pairs $W_h = (\omega_i(j), \omega_{i'}(j'))$ 
of cycles in the increasing order of $o_i(j)$, 
where a tie is resolved arbitrarily. 
We use the concepts and notations above, 
let $y_h$ and $z_h$ be the position of $r_{i'}$ in 
$\tau_{i'}(j')$ at $o_i(j)$ and its normalization, respectively. 
Let $\mathcal{D}_h$ be the set of $z_h$'s such that $y_h$ is visible from $\pi_i(j)$, 
i.e., $r_{i'}$ is visible from $r_i$ at $o_i(j)$. 
Note that $\mathcal{D}_h$ depends on some of $z_1, z_2, \ldots, z_{h-1}$ in general. 
Since the system is rigid, 
$E \in \mathcal{E}(\Omega, \mathcal{A}, I)$ is uniquely determined by $Z = (z_1, z_2, \ldots)$. 
We use the notation $E(Z)$ to emphasize this fact. 

Suppose that $\lambda(\mathcal{D}_h) > 0$ in $E(Z)$ when $z_i = a_i$ for $i = 1, 2, \ldots , h-1$.
By assumption the value of $z_h$ is chosen uniformly at random from $[0,1]$.
If $z_h$ randomly chooses a value in $\mathcal{D}_h$, by Lemma~\ref{lemma1},
regardless of the values chosen for $z_{h+1}, z_{h+2}, \ldots$,
$E(Z) \not\sim \tilde{E} \in \mathcal{E}(\tilde{\Omega}, \mathcal{A}, I)$ 
for any $\tilde{\Omega} \in \SSYNC$ $\lambda$-almost everywhere on $\mathcal{D}_h$,
i.e., $E(Z)$ does not have a similar SSYNC execution $\tilde{E}$ 
for all $z_h \in \mathcal{D}$, except for a countable number of exceptions.  

Thus, if $E$ violates the stationarity, it is unlikely that $E$ has a similar SSYNC execution
(even under rigid system).
In the following, we then assume that $E$ is stationary, and investigate the non-rigid system.

\subsubsection{Pairwise alignment} 

To treat the non-rigid system,
we use an infinite product measure space
$([0,1]^{\infty}, \mathcal{B}^{\infty}([0,1]), \lambda^{\infty})$.
Here $[0,1]^{\infty}$ 
is the Cartesian product of
a countable infinity of copies of $[0,1]$.
The family of events $\mathcal{B}^{\infty}([0,1])$ is the Cartesian product of
a countable infinity of copies of the Borel $\sigma$-field $\mathcal{B}([0,1])$,
which is the $\sigma$-algebra generated by all sets in $[0,1]^{\infty}$ represented by 
a finite union of cylinders, 
where a cylinder is a set in $[0,1]^{\infty}$ with a form 
$B_1 \times B_2 \times \ldots B_n \prod_{i = n+1}^{\infty} [0,1]$
for some natural number $n$ and $B_i \in \mathcal{B}([0,1])$ for all $i = 1, 2, \ldots n$.
Finally $\lambda^{\infty}(\cdot)$ is the product measure on 
$([0,1]^{\infty}, \mathcal{B}^{\infty}([0,1]))$,
which is defined by $\lambda^{\infty}(B) = \prod_{i = 1}^{\infty} \lambda(B_i)$,
for all $B = B_1 \times B_2 \times \ldots \in \mathcal{B}^{\infty}([0,1])$.\footnote{
We can construct $\lambda^{\infty}$ by the Kolmogorov extension theorem 
in the same manner as Theorem 2.4.4 of the book~\cite{T11} by Tao.}

Let $\Phi$ be a property (i.e., predicate) on $[0,1]^{\infty}$ 
and 
let $\Gamma = \{X \mid \Phi(X) \ \text{is true}\} \in \mathcal{B}^{\infty}([0,1])$. 
If $\mathcal{D} \setminus \Gamma$ is a $\lambda^{\infty}$-null set for 
$\mathcal{D} \in \mathcal{B}^{\infty}([0,1])$, 
we say that $\Phi$ holds {\em $\lambda^{\infty}$-almost everywhere} 
on $\mathcal{D}$, 
which means that $\Phi$ holds with probability $1$ under $\lambda^{\infty}$ on $\mathcal{D}$. 

We associate a random variable $y_m$ with the $m$th dimension of $[0,1]^{\infty}$. 
If in $\mathcal{D} \in \mathcal{B}^{\infty}([0,1])$, 
there is a finite set $\{i_1, i_2, \ldots, i_k\} \subset {\mathbf N}$ 
such that $y_{i_k}$ is uniquely determined 
by $y_{i_1}, y_{i_2}, \ldots, y_{i_k-1}$, then $\lambda^{\infty}(\mathcal{D}) = 0$. 

Let $\Omega \in \ASYNC$ be any schedule and 
consider $E \in \mathcal{E}(\Omega, \mathcal{A}, I)$. 
Suppose that in cycle $\omega_i(j)$, 
a robot $r_i$, which is located at $\pi_i(j)$ (in $\mathcal{Z}_0$) at time $o_i(j)$, 
takes a snapshot $P_i(j)$ (in $\mathcal{Z}_i$), 
computes a route $\tau_i(j)$ (in $\mathcal{Z}_0$) by $\mathcal{A}$, 
and moves along $\tau_i(j)$ in its Move. 
Since the system is not rigid, 
it may stop en route after tracing an initial part $\hat{\tau}_i(j)$ 
of $\tau_i(j)$ such that $|\hat{\tau}_i(j)| \geq \delta$. 
Then, the end point of $\hat{\tau}_i(j)$ is $\pi_i(j+1)$ 
of $r_i$ when $\omega_i(j+1)$ starts. 
Thus, $\hat{\tau}_i(j)$ may affect the rest of execution 
including $\pi_i(j+1)$ and $\tau_i(j+1)$. 

The initial part $\hat{\tau}_i(j)$ can be denoted by a real number 
$z_i(j) \in [0,1]$, i.e., 
$z_i(j)$ represents $\hat{\tau}_i(j)$ such that 
$|\hat{\tau}_i(j)| = |\tau_i(j)| z_i(j) - \delta(z_i(j)-1)$. 
Here $z_i(j)$ is well-defined, 
since $\tau_i(j)$ is a simple curve and $|\tau_i(j)| > \delta$.
Since the system is non-rigid, any value $z_i(j) \in [0,1]$ can occur, as assumed. 
More carefully, each of $z_i(j)$'s takes a value in $[0,1]$ uniformly at random,
and the probability that $z_i(j) \in B$ is $\lambda(B)$ for any $B \in \mathcal{B}([0,1])$.

We order $z_i(j)$ in the increasing order of $o_i(j)$, 
where a tie is broken arbitrarily, 
and fix the ordering. 
By associating the $m$th $z_i(j)$ with the $m$th variable $z^{(m)}$, 
we identify $Z = \{z_i(j) \mid r_i \in \mathcal{R}, j \in {\mathbf N}\}$ 
with an infinite vector $(z^{(1)}, z^{(2)}, \ldots) \in [0,1]^{\infty}$.
Then, $E$ is determined uniquely by $Z \in [0,1]^{\infty}$ 
since $E$ is stationary. 
We use notation $E(Z)$ to emphasize that $E$ is determined by $Z$. 
It is easy to observe that $Z = Z'$ if and only if $E(Z) \sim E(Z')$. 

Suppose that $E(Z) \in \mathcal{E}(\Omega, \mathcal{A}, I)$ contains a 
unique triple of cycles $\omega_i(j)$, $\omega_{i'}(j')$, 
and $\omega_{i'}(j'+1)$ such that 
$\omega_i(j)$ and $\omega_{i'}(j'+\ell)$ overlap each other 
and $r_i$ and $r_{i'}$ are mutually visible 
for $\ell = 0,1$.  
Thus, $E(Z)$ does not satisfy the pairwise alignment condition. 
Let $\mathcal{D}$ be the set of $Z$'s such that $E(Z)$ does not satisfy 
the pairwise alignment condition. 
We assume $\lambda^{\infty}(\mathcal{D}) > 0$ by the same reason 
we mentioned immediately above Lemma~\ref{lemma1}.

\begin{lemma}
\label{lemma2} 
Suppose that $\lambda^{\infty}(\mathcal{D}) > 0$. 
Then, $E(Z)$ does not have a similar execution 
$\tilde{E} \in \mathcal{E}(\tilde{\Omega}, \mathcal{A}, I)$ 
for any $\Omega \in \SSYNC$ $\lambda^{\infty}$-almost everywhere on $\mathcal{D}$. 
\end{lemma}
\begin{proof}
Let $\Gamma$ be the set of vectors $Z \in \mathcal{D}$ 
such that $E(Z)$ has a similar SSYNC execution $\tilde{E}$. 
We show that $\lambda^{\infty}(\Gamma) = 0$. 

Suppose that $E(Z)$ has a similar SSYNC execution 
$\tilde{E} \in \mathcal{E}(\tilde{\Omega}, \mathcal{A}, I)$ 
for some $\tilde{\Omega} \in \SSYNC$. 
Let $k = \tilde{o}_i(j)$, $k' = \tilde{o}_{i'}(j')$, 
and $k'' = \tilde{o}_{i'}(j'+1)$, 
where $k' < k''$. 
We have the following four cases: 
(i) $k' < k \leq k''$, 
(ii) $k' \leq k < k''$, 
(iii) $k' < k'' < k$, and 
(iv) $k < k' < k''$. 

\noindent{\bf When $k' < k \leq k''$.~}
At $o_i(j)$, $r_i$ is at $\pi_i(j)$ and $r_{i'}$ is at $\pi_{i'}(j')$ in $E(Z)$, 
and at $k = \tilde{o}_i(j)$, 
$r_i$ is at $\tilde{\pi}_i(j) (=\pi_i(j))$ and 
$r_{i'}$ is at $\tilde{\pi}_{i'}(j'+1) (= \pi_{i'}(j'+1))$ in $\tilde{E}(Z)$. 
Since $P_i(j) = \tilde{P}_i(j)$, $\pi_{i'}(j') \neq \pi_{i'}(j'+1)$, 
$E(Z)$ does not have a similar SSYNC execution $E$, 
since $r_{i'}$ moves at least distance $\delta$ and 
$\tau_{i'}(j')$ is simple. 

\noindent{\bf When $k' \leq k < k''$.~} 
By the argument above, $E(Z)$ does not have a similar SSYNC execution 
$\tilde{E}$ unless $k'=k$. 
Let $\ell \geq 1$ be the minimum integer such that 
$\tilde{o}_i(j+\ell) \geq k''$. 
At $o_{i'}(j'+1)$, $r_i$ is at $\pi_i(j)$ and 
$r_{i'}$ is at $\pi_{i'}(j'+1)$ in $E(Z)$, and 
at $k'' = \tilde{o}_{i'}(j'+1)$, 
$r_i$ is at $\tilde{\pi}_i(j+\ell) (= \pi_i(j+\ell))$ and 
$r_{i'}$ is at $\tilde{\pi}_{i'}(j'+1) (= \pi_{i'}(j'+1))$ in $E(Z)$. 
Since $P_{i'}(j'+1) = \tilde{P}_{i'}(j'+1)$, 
$\pi_i(j) = \pi_i(j+\ell)$. 
That is, $z_i(j+\ell-1)$ is uniquely determined by $\pi_i(j+\ell-1)$. 
(If there is no such $z_i(j+\ell-1)$, $E(Z)$ does not have a similar 
SSYNC execution $\tilde{E}$.) 
Thus, $\lambda^{\infty}(\Gamma) = 0$.  

\noindent{\bf When $k' < k'' < k$.~}
Let $\ell$ be the minimum integer such that 
$\tilde{o}_{i'}(j'+ 1 + \ell) \geq k$. 
By the same argument above, we have $\pi_{i'}(j') = \pi_{i'}(j'+1+\ell)$. 
That is, $z_{i'}(j' + \ell)$ must be uniquely determined by $\pi_{i'}(j'+\ell)$. 
(If there is no such $z_{i'}(j'+\ell)$, $E(Z)$ does not have a 
similar SSYNC execution $\tilde{E}$.) 
Thus, $\lambda^{\infty}(\Gamma) = 0$. 

\noindent{\bf When $k < k' < k''$.~}
By the same argument above, we have $\lambda^{\infty}(\Gamma) = 0$. 

Thus, the proof completes. 
\qed 
\end{proof} 

In order for $Z$ to be in $\mathcal{D}$, 
$E(Z)$ needs to satisfy both $r_{i'} \in S_i(j)$ and $r_i \in S_{i'}(j'+1)$, 
or equivalently, $dist(\pi_i(j), \pi_{i'}(j')) \leq 1$ and 
$dist(\pi_i(j), \pi_{i'}(j'+1)) \leq 1$. 
However, unlike the proof for stationarity, 
$\lambda^{\infty}(\mathcal{D}) > 0$ does not follow in general, 
since there may be a pair of an algorithm $\mathcal{A}$ and 
an initial configuration $I$ 
such that for any $Z \in [0,1]^{\infty}$, 
in $E(Z)$, either $r_{i'} \not\in S_i(j)$ or 
$r_i \not\in S_{i'}(j'+1)$ holds. 
If $\lambda^{\infty}(\mathcal{D}) = 0$, $E(Z)$ satisfies the pairwise alignment 
condition with probability $1$, 
and we do not consider $\mathcal{E}(\Omega, \mathcal{A}, I)$ as an instance 
that does not satisfy the pairwise alignment condition, 
even though there is a $Z \in [0,1]^{\infty}$ such that 
$r_{i'} \in S_{i(j)}$ and $r_i \in S_{i'}(j'+1)$ in $E(Z)$. 

Note that the same claim as Lemma~\ref{lemma2} holds for $E(Z)$ 
such that there are more than one triple $\omega_i(j)$, $\omega_{i'}(j')$, 
and $\omega_{i'}(j'+1)$ 
that violates the pairwise alignment condition. 
Suppose that $E(Z) \in \mathcal{E}(\Omega, \mathcal{A}, I)$ contains 
a pair of cycles $\omega_i(j)$ and $\omega_{i'}(j')$ 
such that $\omega_i(j) \parallel \omega_{i'}(j')$. 
Let $\mathcal{D}$ be the set of $Z$'s such that $E(Z)$ satisfies this condition. 

\begin{proposition}
\label{prop6}
Suppose that $\lambda^{\infty}(\mathcal{D}) > 0$. 
If $E(Z)$ has a similar execution 
$\tilde{E} \in \mathcal{E}(\Omega, \mathcal{A}, I)$ 
for some $\tilde{\Omega} \in \SSYNC$, 
then $\tilde{o}_i(j) = \tilde{o}_{i'}(j')$ 
$\lambda^{\infty}$-almost everywhere on $\mathcal{D}$.  
\end{proposition}
\begin{proof}
Let $\Gamma$ be the set of vectors $Z \in \mathcal{D}$ such that 
there is a $\tilde{\Omega} \in \SSYNC$ satisfying 
\begin{enumerate}
\item 
$\tilde{o}_i(j) \neq \tilde{o}_{i'}(j')$, and 
\item 
$E(Z) \sim \tilde{E} \in \mathcal{E}(\Omega, \mathcal{A}, I)$. 
\end{enumerate}
Then, by the same argument in the proof of Lemma~\ref{lemma2}, 
we have $\lambda^{\infty}(\Gamma) = 0$. 
\qed 
\end{proof}

\subsubsection{Consistency} 

Throughout this section, we assume that $E(Z) \in \mathcal{E}(\Omega, \mathcal{A}, I)$ 
contains a pair of cycles $\omega_i(j)$ and $\omega_{i'}(j')$ 
such that $\omega_i(j) \stackrel{*}{\parallel} \omega_{i'}(j')$, 
and that $E(Z)$ has a similar execution 
$\tilde{E} \in \mathcal{E}(\tilde{\Omega}, \mathcal{A}, I)$ 
for some $\tilde{\Omega} \in \SSYNC$. 
Let $\mathcal{D}$ be the set of $Z$'s such that $E(Z)$ satisfies this condition. 

\begin{lemma}
\label{lemma3} 
Suppose that $\lambda^{\infty}(\mathcal{D}) > 0$. 
If $E(Z) \in \mathcal{E}(\Omega, \mathcal{A}, I)$ which contains a pair 
of cycles $\omega_i(j)$ and $\omega_{i'}(j')$ such that 
$\omega_i(j) \stackrel{*}{\parallel} \omega_{i'}(j')$ holds has a similar 
execution $\tilde{E} \in \mathcal{E}(\tilde{\Omega}, \mathcal{A}, I)$ 
for some $\tilde{\Omega} \in \SSYNC$, 
then $\tilde{o}_i(j) = \tilde{o}_{i'}(j')$ 
$\lambda^{\infty}$-almost everywhere on $\mathcal{D}$. 
\end{lemma}
\begin{proof}
Since $\omega_i(j) \stackrel{*}{\parallel} \omega_{i'}(j')$, 
there are cycles $\omega_{i_{\ell}}(j_{\ell})$ ($\ell = 0, 1, \ldots, m$) 
such that $(i_0, j_0) = (i,j)$ and $(i_m, j_m) = (i', j')$, 
and $\omega_{i_{\ell-1}}(j_{\ell-1}) \parallel \omega_{i_{\ell}}(j_{\ell})$ 
for all $\ell = 1, 2, \ldots, m$. 
By Proposition~\ref{prop6}, 
if $E(Z)$ has a similar SSYNC execution $\tilde{E}$, 
then, for $\ell = 1, 2, \ldots, m$, 
$\tilde{o}_{i_{\ell-1}}(j_{\ell-1}) = \tilde{o}_{i_{\ell}}(j_{\ell})$ 
$\lambda^{\infty}$-almost everywhere on $\mathcal{D}$, 
which implies that $\tilde{o}_i(j) = \tilde{o}_{i'}(j')$ 
$\lambda^{\infty}$-almost everywhere on $\mathcal{D}$. 
\qed
\end{proof}

\begin{lemma}
\label{lemma4} 
Suppose that $\lambda^{\infty}(\mathcal{D}) > 0$. 
If $E(Z)$ has a similar execution 
$\tilde{E} \in \mathcal{E}(\tilde{\Omega}, \mathcal{A}, I)$ 
for some $\tilde{\Omega} \in \SSYNC$, 
then $r_{i'} \in S_i(j)$ if and only if $r_i \in S_{i'}(j')$, 
$\lambda^{\infty}$-almost everywhere on $\mathcal{D}$. 
\end{lemma}

\begin{proof}
Since $\lambda^{\infty}(\mathcal{D}) > 0$, 
$E(Z)$ has a similar execution 
$\tilde{E} \in \mathcal{E}(\tilde{\Omega}, \mathcal{A}, I)$ 
for some $\tilde{\Omega} \in \SSYNC$ such that 
$\tilde{o}_i(j) = \tilde{o}_{i'}(j')$ 
$\lambda^{\infty}$-almost everywhere on $\mathcal{D}$ 
by Lemma~\ref{lemma3}. 

Suppose that $o_{i'}(j') < o_i(j)$. 
Let $\omega_{i'}(j'+ \ell'-1)$ (resp. $\omega_i(j-\ell-1)$) be the 
cycle of $r_{i'}$ (resp. $r_i$) such that 
$\pi_{i'}(j'+\ell')$ (resp. $\pi_i(j+\ell)$) 
be the position of $r_{i'}$ (resp. $r_i$) 
at $o_i(j)$ (resp. $o_{i'}(j')$). 
Since $\pi_i(j) = \tilde{\pi}_i(j)$, 
$\pi_i(j-\ell) = \tilde{\pi}_i(j-\ell)$, 
$\pi_{i'}(j') = \tilde{\pi}_{i'}(j')$, 
$\pi_{i'}(j'+\ell) = \tilde{\pi}_{i'}(j'+\ell)$, 
$\pi_i(j) = \tilde{\pi}_i(j-\ell)$, 
and $\pi_{i'}(j') = \tilde{\pi}_{i'}(j' + \ell)$, 
provided that $\tilde{o}_i(j) = \tilde{o}_{i'}(j')$. 
Obviously, the set of $Z \in \mathcal{D}$ satisfying 
$\pi_i(j) = \pi_i(j-\ell)$ (resp. $\pi_{i'}(j') = \pi_{i'}(j' + \ell')$) 
is the $\lambda^{\infty}$-null set when $\ell > 0$ (resp. $\ell' > 0$). 
Thus, $r_{i'} \in S_i(j)$ if and only if $r_i \in S_{i'}(j')$ 
$\lambda^{\infty}$-almost everywhere on $\mathcal{D}$. 

The case $o_{i'}(j') > o_i(j)$ is analogous 
and the case $o_{i'}(j') = o_i(j)$ is trivial. 
\qed 
\end{proof}

By the proof of above lemma, we have the following corollary. 

\begin{corollary}
\label{corl2} 
Suppose that $\lambda^{\infty}(\mathcal{D}) > 0$. 
If $E(Z)$ has a similar execution 
$\tilde{E} \in \mathcal{E}(\tilde{\Omega}, \mathcal{A}, I)$ 
for some $\tilde{\Omega} \in \SSYNC$, 
and $r_{i'} \in S_i(j)$ and $r_i \in S_{i'}(j')$ hold, 
then $\omega_i(j) \parallel \omega_{i'}(j')$ 
$\lambda^{\infty}$-almost everywhere on $\mathcal{D}$. 
\end{corollary}

\begin{lemma}
\label{lemma5} 
Suppose that $\lambda^{\infty}(\mathcal{D}) > 0$. 
If $E(Z)$ has a similar execution 
$\tilde{E} \in \mathcal{E}(\tilde{\Omega}, \mathcal{A}, I)$ 
for some $\tilde{\Omega} \in \SSYNC$, 
and $r_{i'} \not\in S_i(j)$ and $r_i \not\in S_{i'}(j')$ hold, 
then $dist(\pi_i(j), \pi_{i'}(j')) > 1$ 
$\lambda^{\infty}$-almost everywhere on $\mathcal{D}$. 
\end{lemma}
\begin{proof}
Since $\lambda^{\infty}(\mathcal{D}) > 0$, 
$E(Z)$ has a similar execution 
$\tilde{E} \in \mathcal{E}(\tilde{\Omega}, \mathcal{A}, I)$ 
for some $\tilde{\Omega} \in \SSYNC$, 
such that $\tilde{o}_i(j) = \tilde{o}_{i'}(j')$ 
$\lambda^{\infty}$-almost everywhere on $\mathcal{D}$ 
by Lemma~\ref{lemma3}. 

If $dist(\pi_i(j), \pi_{i'}(j')) \leq 1$, 
since $\pi_i(j) = \tilde{\pi}_i(j)$, $\pi_{i'}(j') = \tilde{\pi}_{i'}(j')$, 
and $\tilde{o}_i(j) = \tilde{o}_{i'}(j')$ 
$\lambda^{\infty}$-almost everywhere on $\mathcal{D}$, 
there is a robot $r_{i''}$ and cycle $\omega_{i''}(j'')$ 
such that $r_{i''}$ is at position $\pi_{i'}(j')$ at time $o_i(j)$, 
since $P_i(j) = \tilde{P}_i(j)$. 
Such event does not occur 
$\lambda^{\infty}$-almost everywhere on $\mathcal{D}$. 
Thus, $dist(\pi_i(j), \pi_{i'}(j')) > 1$ 
$\lambda^{\infty}$-almost everywhere on $\mathcal{D}$. 
\qed 
\end{proof}

\subsubsection{Serializability} 

\begin{lemma}
 \label{lemma6} 
If $E(Z) \in \mathcal{E}(\Omega, \mathcal{A}, I)$ which contains 
a pair of cycles $\omega_i(j)$ and $\omega_{i'}(j')$ such that 
$\omega_i(j) \rightarrow \omega_{i'}(j')$ holds 
has a similar execution $\tilde{E} \in \mathcal{E}(\Omega, \mathcal{A}, I)$ 
for some $\tilde{\Omega} \in \SSYNC$, 
then $\tilde{o}_i(j) < \tilde{o}_{i'}(j')$ 
$\lambda^{\infty}$-almost everywhere on $\mathcal{D}$. 
\end{lemma}
\begin{proof}
If $i = i'$, since $j < j'$, $\tilde{o}_i(j) < \tilde{o}_{i'}(j')$ by definition. 

Suppose that $i \neq i'$. 
Then, $r_i \in S_{i'}(j')$ and $o_{i'}(j') \in (f_i(j), o_i(j+1))$. 
Then, by a similar argument in the proof of Lemma~\ref{lemma2}, 
$E(Z)$ does not have a similar execution $\tilde{E} \in \mathcal{E}(\Omega, \mathcal{A}, I)$ for any $\tilde{\Omega} \in \SSYNC$ such that 
$\tilde{o}_i(j) > \tilde{o}_{i'}(j')$ 
$\lambda^{\infty}$-almost everywhere on $\mathcal{D}$. 
\qed 
\end{proof}

The following corollary immediately holds by Lemma~\ref{lemma3} and 
Lemma~\ref{lemma6}. 
\begin{corollary}
If $E(Z)$ does not satisfy the serializability, 
then $E(Z)$ does not have a similar execution 
$\tilde{E} \in \mathcal{E}(\tilde{\Omega}, \mathcal{A}, I)$ 
for any $\tilde{\Omega} \in \SSYNC$ 
$\lambda^{\infty}$-almost everywhere on $\mathcal{D}$. 
\end{corollary}

\subsubsection{Naturality} 

Recall that $\mathcal{G} = (\{\Omega_0, \Omega_1, \ldots\}, \Rightarrow)$ 
and the set $\mathcal{T}$ of topological sorts of $\mathcal{G}$ are 
determined by execution $E(Z) \in \mathcal{E}(\Omega, \mathcal{A}, I)$. 
We sometimes associate $E(Z)$ with $\mathcal{G}$ and $\mathcal{T}$ 
to emphasize that they are uniquely determined by $E(Z)$. 
Let $\mathcal{TS}(E(Z)) (\subset \SSYNC)$ be the set of schedules 
constructed from topological sorts in $\mathcal{T}(E(Z))$. 
By Lemma~\ref{lemma3} and Lemma~\ref{lemma6}, 
if $\tilde{E} \in \mathcal{E}(\tilde{\Omega}, \mathcal{A}, I)$ 
is similar to $E(Z) \in \mathcal{E}(\Omega, \mathcal{A}, I)$, 
then $\tilde{\Omega} \in \mathcal{TS}(E(Z))$ 
$\lambda^{\infty}$-almost everywhere on $\mathcal{D}$, 
provided that $\lambda^{\infty}(\mathcal{D}) > 0$. 
Here $\mathcal{D}$ is the set of $U \in [0,1]^{\infty}$ 
such that $\mathcal{G}(E(U)) = \mathcal{G}(E(Z))$. 

Consider any schedule $\tilde{\Omega} \in \mathcal{TS}(E(Z))$. 
Suppose that there is a pair of cycles $\omega_i(j)$ and $\omega_{i'}(j')$ 
such that $k' \leq k < k''$, 
where, under $\tilde{\Omega}$, 
$\tilde{o}_i(j) = k$, $\tilde{o}_{i'}(j'-1) = k'$, and 
$\tilde{o}_{i'}(j') = k''$. 
Obviously, $\omega_i(j) \not\stackrel{*}{\parallel} \omega_{i'}(j')$. 

First assume that $r_{i'} \in S_i(j)$ in $E(Z)$, and 
let $\mathcal{D'} \subseteq \mathcal{D}$ be the set of $Z \in \mathcal{D}$ 
such that $E(Z)$ satisfies this condition. 

\begin{lemma}
\label{lemma7} 
Suppose that $\lambda^{\infty}(\mathcal{D'}) > 0$. 
If $E(Z)$ has a similar  execution 
$\tilde{E} \in \mathcal{E}(\Omega, \mathcal{A}, I)$, 
then $o_i(j) < o_{i'}(j')$ 
$\lambda^{\infty}$-almost everywhere on $\mathcal{D'}$.  
\end{lemma}
\begin{proof}
By definition, $r_{i'}$ is at position $\tilde{\pi}_{i'}(j')$ 
at time $\tilde{o}_i(j) = k$ in $\tilde{E}$. 
Since  $\omega_i(j) \not\stackrel{*}{\parallel} \omega_{i'}(j')$, 
$o_i(j) \neq o_{i'}(j')$. 
Suppose that $o_i(j) > o_{i'}(j')$. 
There is a cycle $\omega_{i'}(j'+\ell')$ such that $r_{i'}$ 
is at $\pi_{i'}(j' + \ell')$ at $o_i(j)$ for some $\ell' \geq 1$. 
Since $\pi_i(j) = \tilde{\pi}_i(j)$, $P_i(j) = \tilde{P}_i(j)$, 
and $r_{i'} \in S_i(j)$, there is a robot $r_{i''}$ 
(which may be $r_{i'}$) and a cycle $\omega_{i''}(j'')$ such that 
$\tilde{\pi}_{i''}(j'') = \pi_{i'}(j' + \ell')$. 
Thus, $E(Z)$ does not have a similar execution 
$\tilde{E} \in \mathcal{E}(\tilde{\Omega}, \mathcal{A}, I)$ 
$\lambda^{\infty}$-almost everywhere on $\mathcal{D'}$. 
\qed 
\end{proof}

Next assume that $r_{i'} \not\in S_i(j)$ in $E(Z)$, 
and let $\mathcal{D}'' \subseteq \mathcal{D}$ be the set of $Z \in \mathcal{D}$ 
such that $E(Z)$ satisfies this condition. 

\begin{lemma}
\label{lemma8} 
Suppose that $\lambda^{\infty}(\mathcal{D''}) > 0$. 
If $E(Z)$ has a similar  execution 
$\tilde{E} \in \mathcal{E}(\Omega, \mathcal{A}, I)$, 
then $dist(\pi_i(j), \pi_{i'}(j')) > 1$ 
$\lambda^{\infty}$-almost everywhere on $\mathcal{D''}$. 
\end{lemma}
\begin{proof}
Suppose that $dist(\pi_i(j), \pi_{i'}(j')) \leq 1$. 
Since $r_{i'}$ is at $\tilde{\pi}_{i'}(j')$ at $\tilde{o}_i(j) = k$ 
in $\tilde{E}$, 
$\tilde{\pi}_{i'}(j') = \pi_{i'}(j')$ and $\tilde{\pi}_i(j) = \pi_i(j)$, 
$r_{i'} \in \tilde{S}_i(j)$. 
There is a robot $r_{i''}$ for some $i''(\neq i)$ and a cycle 
$\omega_{i''}(j'')$ such that $r_{i''}$ is at $\pi_{i'}(j')$ at $o_i(j)$ in $E(Z)$, 
since $\tilde{P}_i(j) = P_i(j)$. 
Thus, $E(Z)$ does not have a similar execution 
$\tilde{E} \in \mathcal{E}(\tilde{\Omega}, \mathcal{A}, I)$ 
$\lambda^{\infty}$-almost everywhere on $\mathcal{D''}$. 
\qed 
\end{proof}

Consequently, we have the following theorem. 

\begin{theorem}
\label{theorem:necessity} 
Each of the five properties, 
stationarity, pairwise alignment, consistency, serializability, and naturality 
is necessary for an execution $E(Z) \in \mathcal{E}(\Omega, \mathcal{A}, I)$ to have 
a similar execution $\tilde{E} \in \mathcal{E}(\tilde{\Omega}, \mathcal{A}, I)$ 
for some schedule $\tilde{\Omega} \in \SSYNC$ with probability $1$. 
\end{theorem}

\section{Luminous synchronizer for ASYNC robots} 
\label{sec:synchronizer} 

In this section, we present a synchronizer for oblivious luminous ASYNC robots, 
that produces ASYNC executions satisfying 
Assumptions~\ref{ass1}, \ref{ass2}, \ref{ass3}, \ref{ass4}, and \ref{ass5}. 
When the robots are not equipped with lights, 
each robot cannot recognize which robots are moving. 
We compensate for this weak capability by a single light 
at each robot.\footnote{
Remember that the color of a light is changed at the end of a Compute, 
and it is kept until the end of the Compute of the next cycle.} 

Let $C$ be the set of colors that a light can take. 
Each light initially takes Black ($Bk \in C$). 
When robot $r_i$ takes a snapshot $Q_i$ at time $t$, 
$Q_i$ is the set of pairs $(p_{i'}, c_{i'})$ 
for each $r_{i'}$ visible for $r_i$ at $t$, 
where $p_{i'}$ is the the position of $r_{i'}$ at $t$ in $\mathcal{Z}_i$ 
and $c_{i'}$ is the color of $r_{i'}$'s light at $t$. 
Let $P(Q_i) = \{p_{i'} \mid (p_{i'}, c_{i'}) \in Q_i\}$, 
and $C(Q_i) = \{c_{i'} \mid (p_{i'}, c_{i'}) \in Q_i\}$. 
That is, $P(Q_i)$ is the set of positions occupied by robots 
visible for $r_i$ in $\mathcal{Z}_i$, 
and $C(Q_i)$ is the set of colors visible for $r_i$. 
Recall that $r_i$ is aware of its color, i.e., the color of its light $c_i$ 
since $p_i=(0,0)$ and the robots occupy distinct points. 

For an initial configuration $I$ of the system of non-luminous robots, 
let $\hat{I} = \{(p, Bk) \mid p \in I\}$ be an initial configuration 
of the system of luminous robots. 
By definition, $P(\hat{I}) = I$ and $C(\hat{I}) = \{Bk\}$. 

We now define a luminous 
synchronizer $\mathcal{S}$ on a robot $r_i$ under any schedule 
$\Omega \in \ASYNC$. 
Given any algorithm $\mathcal{A}$ and initial configuration $I$ for non-luminous robots, 
the initial configuration for $\mathcal{S}$ is the corresponding configuration $\hat{I}$. 
Luminous synchronizer 
$\mathcal{S}$ on $r_i$ inhibits {\em on-the-fly} $r_i$'s motion 
in some cycle so that the resulting execution have a similar SSYNC execution. 
Precisely, $\mathcal{S}$ works as follows: 
In a cycle $\omega = (o, s, f)$ of robot $r$, 
$r$ takes a snapshot $Q$ at time $o$ in Look. 
In Compute, depending on $Q$, $\mathcal{S}$ on $r$ first decides whether or not it 
``accepts'' $\omega$, and then computes a move route $\tau$. 
If it accepts $\omega$, 
$\tau$ is the one that $\mathcal{A}$ computes given $P(Q)$;
otherwise, if it ``rejects'' $\omega$, 
$\tau$ is the point $(0,0)$. 
Finally, it decides a color $c \in C$ and 
the color of $r_i$'s light is changed to $c$ at the end of Compute. 
Thus, the color $c$ is visible from other robots at $s$ and thereafter. 
In Move, $r$ traces $\tau$ but it may stop en route after moving distance $\delta$. 

The set of executions $F$ of $\mathcal{S}$ for $\mathcal{A}$ and $I$ under scheduler $\Omega$ 
is denoted by $\mathcal{E}(\Omega, \mathcal{S}(\mathcal{A}), \hat{I})$. 
Let $\Lambda$ be the set of cycles accepted by $\mathcal{S}$ in $F$. 
Note that $\Lambda$ depends on $F$. 
From $F$ (and $\Lambda$), we can construct an execution 
$\check{F} \in \mathcal{E}(\Lambda, \mathcal{A}, I)$ for non-luminous robots 
by first extracting the behaviors of the robots for 
cycles in $\Lambda$ and then ignoring the colors of lights. 
Since the next position is computed from $P(Q)$ (not from $Q$) 
and the robots do not change their positions in rejected cycles, 
indeed $\check{F} \in \mathcal{E}(\Lambda, \mathcal{A}, I)$. 

We say that luminous 
synchronizer $\mathcal{S}$ is {\em correct} if the following conditions 
hold for any $\Omega$, $\mathcal{A}$, $I$, and $F \in \mathcal{E}(\Omega, \mathcal{S}(\mathcal{A}), \hat{I})$. 
\begin{enumerate}
\item 
$\Lambda$ is fair. 
\item 
$\check{F}$ satisfies Assumptions~\ref{ass1}, \ref{ass2}, \ref{ass3}, 
\ref{ass4}, and \ref{ass5}, which implies that 
$\check{F} \in \mathcal{E}(\Lambda, \mathcal{A}, I)$ has a similar execution 
$\tilde{E} \in \mathcal{E}(\tilde{\Lambda}, \mathcal{A}, I)$ for some 
$\tilde{\Lambda} \in \SSYNC$. 
\end{enumerate}

\subsection{Limit of color-based synchronizer} 

The {\em visibility graph} of a configuration of the robots 
consists of a set of vertices corresponding to the robots 
and a set of edges between any pair of robots 
within distance $1$ (in $\mathcal{Z}_0$). 
A {\em visibility preserving} algorithm guarantees that 
the visibility graph does not change in any execution. 
Formally, 
an execution $E \in \mathcal{E}(\Omega, \mathcal{A}, I)$ for non-luminous robots 
is {\em visibility preserving}, 
if the following condition holds: 
For any $r_i$ and $r_{i'}$, and for any time $t \in {\mathbf R}^+$, 
$dist(p_i(t), p_{i'}(t)) \leq 1$ if and only if 
$dist (p_i(0), p_{i'}(0)) \leq 1$, 
where $p_j(u)$ is the position of $r_j$ at time $u$ in $\mathcal{Z}_0$. 
We say that an algorithm $\mathcal{A}$ is {\em visibility preserving}, 
if every execution $E \in \mathcal{E}(\Omega, \mathcal{A}, I)$ is 
visibility preserving for any $\Omega \in \SSYNC$ and $I$. 

Let $Q_i(j) = \{(p_{i'}, c_{i'}) \mid r_{i'} \in S_i(j) \}$ be the 
snapshot taken by a robot $r_i$ at $o_i(j)$ in $\omega_i(j)$. 
Clearly, $p_i = (0,0)$. 
Let $X_i(j) = \{ c_{i'} \mid i' \neq i, r_{i'} \in S_i(j)\}$. 
In general, a luminous 
synchronizer on $r_i$ can use the full information on $Q_i(j)$ 
to decide whether it accepts $\omega_i(j)$ or not. 
A luminous synchronizer $\mathcal{S}$ is {\em color-based} 
if it uses $c_i$ and $X_i(j)$ for the selection. 
A color-based synchronizer $\mathcal{S}$ is {\em greedy} 
if it accepts $\omega_i(j)$ if and only if $C(Q_i(j)) = \{Bk\}$, i.e., 
$c_i = Bk$ and $X_i(j)$ is either $\{Bk\}$ or $\emptyset$. 
We show that any greedy synchronizer is not powerful enough in general. 

\begin{lemma}
\label{lemma:removal}
There exists a rigid system of five luminous robots such that 
for any greedy synchronizer $\mathcal{S}$, 
there is a triple $(\Omega, \mathcal{A}, I)$ such that, 
for some execution 
$F \in \mathcal{E}(\Omega, \mathcal{S}(\mathcal{A}), \hat{I})$, 
$\check{F}$ is not consistent. 
\end{lemma}
\begin{proof}
Consider a system of five luminous robots. 
We illustrate an initial part of an execution 
$F \in \mathcal{E}(\Omega, \mathcal{S}(\mathcal{A}), I)$. 
Initially, $r_1, r_2, \ldots, r_5$ are at $(0,0)$, $(0,1)$, $(1,1)$, $(1,0)$, 
and $(2,0)$, respectively in $\mathcal{Z}_0$. 
Thus, 
$\hat{I} = \{((0,0), Bk), ((0,1), Bk), ((1,1), Bk), ((1,0), Bk), ((2, 0), Bk)\}$. 
Since the visibility range is $1$, $r_2$ and $r_4$ are visible from $r_1$, 
but $r_5$ and $r_3$ are not visible from $r_1$. 
The first cycles in $\Omega$ are 
$\omega_1(1) = (0, 3/4, 1)$, 
$\omega_2(1) = (1/2, 5/4, 3/2)$, 
$\omega_3(1) = (1, 7/4, 2)$, 
$\omega_4(1) = (3/2, 9/4, 5/2)$, and 
$\omega_5(1) = (3, 15/4, 4)$.  
Recall that the initial color of light is $Bk$. 
Since $\mathcal{S}$ is greedy, it accepts $\omega_1(1)$, $\omega_2(1)$, 
and $\omega_3(1)$. 

Suppose that $\mathcal{A}$ moves $r_1$ to $(0, 3/4)$ in $\omega_1(1)$. 
Then, at time $3/2$, $r_1$ is not visible from $r_4$. 
Thus, all robots visible from $r_4$ is still have color $Bk$, 
and $\mathcal{S}$ on $r_4$ accepts $\omega_4(1)$. 

Observe that $\omega_1(1) \stackrel{*}{\parallel} \omega_4(1)$. 
Then, $\check{F}$ is not consistent regardless of the rest of $F$, 
because $r_4 \in S_1(1)$ but $r_1 \not\in S_4(1)$. 
\qed 
\end{proof}

We extend Lemma~\ref{lemma:removal} to the color-based synchronizer. 
A fully synchronous scheduler $\FSYNC$ produces a schedule $\Omega$ 
such that $\omega_i(j) = (j-1, j-3/4, j-1/4)$ for all $i$ and $j$. 
Thus, $\FSYNC \subset \SSYNC$ in the sense that if $\Omega \in \FSYNC$, 
$\Omega \in \SSYNC$ holds.  

\begin{theorem}
\label{theorem:removal} 
There exists a rigid system of five luminous robots such that 
for any color-based synchronizer $\mathcal{S}$, 
there is a triple $(\Omega, \mathcal{A}, I)$ such that, 
for some execution 
$F \in \mathcal{E}(\Omega, \mathcal{S}(\mathcal{A}), \hat{I})$, 
$\check{F}$ is not consistent. 
\end{theorem}
\begin{proof}
Remember the execution $F \in \mathcal{E}(\Omega, \mathcal{S}(\mathcal{A}), I)$ 
for five luminous robots  in the proof of Lemma~\ref{lemma:removal}. 
The counter example relies on the assumptions 
that $\mathcal{S}$ is greedy and 
that the initial color of each robot is $Bk$. 
We show that there is a $\Lambda$ which changes the configuration to 
a one such that $\mathcal{S}$ on each robot accepts the current cycle. 

Consider an execution 
$F \in \mathcal{E}(\Lambda, \mathcal{S}(\mathcal{A}), \hat{I})$, 
where $\Lambda \in \FSYNC$. 
Then, all lights have the same color $c_j$ 
when the robots simultaneously start their $j$th cycles 
for any $j \in {\mathbf N}$ since $\mathcal{S}$ is color-based. 
Hence, there is a $j_0$ such that each robot $r_i$ accepts 
$\omega_i(j_0)$ since $\mathcal{S}$ is fair. 
That is, $\mathcal{S}$ on each robot $r_i$ accepts $\omega_i(j_0)$, 
since it accepts a cycle when $c = c_{j_0-1}$ and $X=\{C_{j_0-1}\}$. 
We assume without loss of generality that $\omega_i(j_0)$ is the first 
cycle accepted by $\mathcal{S}$. 

Now the configuration is 
$\{(p, c_{j_0-1}) \mid (p, Bk) \in \hat{I} \}$ 
immediately before the $j_0$th cycle starts at time $j_0-1$. 
We replace $\omega_i(j_0)$ for each $r_i \in \Lambda$ with 
\begin{itemize}
\item 
$\omega_1(j_0) = (j_0-1, (j_0-1) + 3/4, (j_0-1)+1)$, 
\item 
$\omega_2(j_0) = ((j_0-1) + 1/2, (j_0-1) + 5/4, (j_0-1) + 3/2)$, 
\item 
$\omega_3(j_0) = ((j_0-1) + 1, (j_0-1) + 7/4, (j_0-1) + 2)$, 
\item 
$\omega_4(j_0) = ((j_0-1) + 3/2, (j_0-1) + 9/4, (j_0-1) + 5/2)$, and 
\item 
$\omega_2(j_0) = ((j_0-1) + 3, (j_0-1) + 15/4, (j_0-1) + 4)$. 
\end{itemize}
Then, the same argument as the proof of Lemma~\ref{lemma:removal} 
concludes the theorem.  
\qed 
\end{proof}

The above example shows that 
there exists no color-based synchronizer that works correctly 
if the algorithm is not visibility preserving. 

\subsection{Color-based synchronizer for vicinity preserving algorithms} 

Lemma~\ref{lemma:removal} and Theorem~\ref{theorem:removal} demonstrate 
that there is no color-based synchronizer for an arbitrary 
visibility preserving algorithm. 
Moreover, it is difficult for oblivious luminous robots 
to satisfy the second condition of naturality, 
because it requires remembering the positions of other robots.
In this section, we consider algorithms 
with more restricted changes in the visibility graph and 
propose a color-based synchronizer for such algorithms. 

An execution $E \in \mathcal{E}(\Omega, \mathcal{A}, I)$ for non-luminous robots 
is {\em vicinity preserving}, 
if the following condition holds: 
For any $r_i$, $r_{i'}$, and for any time $t, t' \in {\mathbf R}^+$, 
$dist(p_i(t), p_{i'}(t')) \leq 1$ if and only if 
$dist (p_i(0), p_{i'}(0)) \leq 1$. 
In other words, $r_i$ has to stay in the vicinity of its initial position. 
We say that an algorithm $\mathcal{A}$ is {\em vicinity preserving}, 
if every execution $E \in \mathcal{E}(\Omega, \mathcal{A}, I)$ is vicinity preserving 
for any $\Omega \in \SSYNC$ and $I$. 
In this section, we propose a color-based synchronizer 
$\mathcal{S}_{VP}$ 
that uses a set 
$C = \{Bk, R, B, G, W\}$ of colors 
and show its correctness provided that $\mathcal{A}$ is vicinity preserving 
and designed for non-luminous SSYNC robots. 
We describe $\mathcal{S}_{VP}$ as a finite-state machine 
with a state set $C$, an input alphabet $2^C$, 
and an output alphabet $\{\text{accept}, \text{reject} \}$. 
When $\mathcal{S}_{VP}$ is executed on $r_i$, 
the sate of $r_i$ is the color (of the light) of $r_i$ 
and the input is the set $X$ of colors of the robots visible from $r_i$, excluding $r_i$'s color. 
Table~\ref{table:S-ST} shows the transition function and 
the output function of $\mathcal{S}_{VP}$, where 
\begin{itemize}
\item 
$\exists c$ means any $X$ such that $c \in X$, 
\item 
$\forall(c_1, c_2, \ldots, c_k)$ means any $X$ such that 
$X \subseteq \{c_1, c_2, \ldots, c_k\}$, and 
\item 
$\exists c \ \wedge \ \forall(c_1, c_2, \ldots, c_k)$ means 
any $X$ such that $c \in X$ and $X \subseteq \{c_1, c_2, \ldots, c_k\}$. 
\end{itemize}
The initial state of each robots is $Bk$. 
We assume that without loss of generality, 
the visibility graph of the initial configuration $I$ is connected. 
Since $\mathcal{A}$ is vicinity preserving, 
let $S_i$ be the set of neighbors of $r_i$ 
in the visibility graph which is defined by $I$. 

\begin{table}[t]
\centering 
\caption{Finite-state machine $\mathcal{S}_{VP}$.}
\begin{tabular}{|c|l|c|c|} \hline
Current state & Input & Next state & Output \\ \hline \hline
Bk & $\forall (Bk,B,W)$ & $R$ & accept \\ \cline{2-4}
      & $\exists R \wedge \forall (Bk,R,B,W)$ & $W$ & reject \\  \hline
R & $\forall (R,B,W)$ & $B$ & reject \\ \hline
B & $\forall (B,G)$ & $G$ & reject \\ \hline
G & $\forall (Bk,G)$ & $Bk$ & reject \\ \hline
W & $\forall (B,W)$ & $Bk$ & reject \\ \hline
\end{tabular}
\label{table:S-ST}
\end{table}

Robot $r_i$ is waiting when its state is $Bk$ (Black) 
and moving when its state is $R$ (Red). 
It rejects the current cycle when its state is $Bk$ and 
observes another robot in state $R$, 
and changes its state to $W$ (White). 
It changes its state from $W$ to $Bk$ when it does not observe any robot in state $R$. 
Robot $r_i$ has finished moving when its state is $B$ (Blue) and 
it changes its state to $G$ (Green) when the states of robots in $S_i$ 
are $B$ and $G$. 
Finally, it changes its state to $Bk$ when the states of robots in $S_i$ are $G$ and $Bk$. 

Let $F \in \mathcal{E}(\Omega, \mathcal{S}_{VP}(\mathcal{A}), \hat{I})$ 
be any execution of $\mathcal{S}_{VP}$ for a vicinity preserving algorithm $\mathcal{A}$, 
and let $\Lambda \subseteq \Omega$ be the set of accepted cycles 
(which depends on $F$). 
Then, $\check{F} \in \mathcal{E}(\Lambda, \mathcal{A}, I)$ is an 
execution of non-luminous robots. 
We will show that $\check{F}$ 
has a corresponding SSYNC execution for non-luminous robots. 
We demonstrate that (i) $\Lambda$ is fair, 
(ii) $\check{F}$ satisfies stationarity, pairwise alignment, 
consistency, serializability, and naturality. 

\begin{lemma}
\label{lemma:fairness} 
$\Lambda$ is fair. 
\end{lemma}
\begin{proof}
We show that every robot reaches state $R$ infinitely many times in 
$F$. 
We first regard three states $Bk$, $W$, and $R$ as a virtual state 
$Y$. 
Every robot starts with state $Y$. 
When the state of a robot $r_i$ is $Y$ (i.e., $R$) and 
the sate of each robot $r_{i'} \in S_i$ is either $Y$ 
(i.e., $R$ or $W$) or $B$, 
then it can change its state to $B$. 
When the state of $r_i$ is $G$ and the state of each robot 
$r_{i'} \in S_i$ is either $G$ or $Y$ (i.e., $Bk$), 
then it can change its state to $Y$ (i.e., $Bk$). 
For the time being, we assume that if the state of $r_i$ is $Bk$, 
then it will eventually change its state to $R$. 
We show that every robot $r_i$ reaches state $Y$ infinitely 
many times in $F$. 
This implies that it reaches state $R$ infinitely many times. 
Then, we can conclude that $\check{F}$ is fair, 
because $r_i$ changes its color to $R$ in $\omega_i(j)$ 
if and only if $\omega_i(j)$ is accepted by $\mathcal{S}_{VP}$. 

Consider any robot $r_i$. 
Let $\Psi_i = (\psi_i(1), \psi_i(2), \ldots) \subseteq \Omega_i$, 
where $\psi_i(j)$ is the cycle of $r_i$ in which 
it changes its state for the $j$th time. 
The state of $r_i$ changes from state $Y$ to $B$ in $\psi_i(1)$, 
from state $B$ to $G$ in $\psi_i(2)$, 
from state $G$ to $Y$ in $\psi_i(3)$, ans so on. 
For each $C \in \{Y, B, G\}$, 
$C^{(k)}$ denotes the the state that $r_i$ takes $C$ for the 
$k$th time. 
Thus, the state of $r_i$ changes from state $Y^{(1)}$ to $B^{(1)}$ in $\psi_i(1)$, 
from state $B^{(1)}$ to $G^{(1)}$ in $\psi_i(2)$, 
from state $G^{(1)}$ to $Y^{(2)}$ in $\psi_i(3)$, and so on. 
Independently of $i$, in $\psi_i(j)$, 
$r_i$ changes its state from $c(j)$ to $c(j+1)$, where 
$c(j) = Y^{(\lfloor j/3 \rfloor +1)}$ if $j \pmod{3} = 1$, 
$c(j) = B^{(\lfloor j/3 \rfloor +1)}$ if $j \pmod{3} = 2$, and 
$c(j) = G^{(\lfloor j/3 \rfloor)}$ if $j \pmod{3} = 0$. 

Let $\psi_i(j) = (o_i(j), s_i(j), f_i(j))$. 
Then, $r_i$ takes a snapshot at $o_i(j)$ 
and changes the color of its light (i.e., its state) by $s_i(j)$. 
The new color becomes visible from other robots at $s_i(j)$. 
Thus, the state of $r_i$ at $o_i(j)$ is $c(j)$ and 
is $c(j+1)$ at $s_i(j)$. 

For any robot $r_{\ell}$, let $\sigma_{\ell}(t)$ be the sate of $r_{\ell}$ 
at time $t \in {\mathbf R}^+$. 
Thus, $\sigma_i(o_i(j)) = c(j)$ and $\sigma_i(s_i(j)) = c(j+1)$. 
We first claim that 
for any $j \in {\mathbf N}$, 
$\sigma_{i'}(o_i(j))$ is either $c(j)$ or $c(j+1)$ for any robot 
$r_{i'} \in S_i$. 
The proof is by induction on $j$. 

When $j=1$, $c_i(o_i(1)) = c(1) = Y^{(1)}$. 
Suppose that $\sigma_{i'}(o_i(j)) = c(j')$ for some $j'\geq 3$. 
Then $r_{i'}$ changes its state from $B^{(1)}$ to $G^{(1)}$ 
in $\psi_{i'}(2)$ and $o_{i'}(2) < o_i(1)$. 
It is a contradiction, because $r_i \in S_{i'}$ and 
$\sigma_i(o_{i'}(2)) = Y^{(1)}$. 
Thus, $\sigma_{i'}(o_i(1))$ is either $c(1) (= Y^{(1)})$ 
or $c(2) (= B^{(1)})$. 

Provided that $\sigma_{i'}(o_i(j))$ is either $c(j)$ or $c(j+1)$, 
we show that $\sigma_{i'}(o_i(j+1))$ is either $c(j+1)$ or $c(j+2)$
. 
By definition, if $\sigma_{i'}(o_i(j+1)) = c(\ell)$, then $\ell \geq j$. 
If $\ell = j$, then $r_i$ cannot change its state from $c(j+1)$ to $c(j+2)$ 
in $\psi_i(j+1)$. 
Thus, $\ell \geq j+1$. 
Suppose that $\ell \geq j+3$. In time interval $[o_i(j), s_i(j))$, 
the state of $r_i$ is still $c_i(j)$. 
During this interval, the state of $r_i'$ is either $c(j)$ or $c(j+1)$. 
Thus, by the same argument as the base case for $\psi_{i'}(j+2)$, 
a contradiction is derived. 

To show that $\Psi_i$ is an infinite sequence, 
we assume that it is a finite sequence and derive a contradiction. 
Let $h$ be the length of $\Psi_i$, 
i.e., $\psi_i(h)$ is the last cycle of $\Psi_i$. 
By the claim above, $\Psi_k$ is finite for any $k$. 
Without loss of generality, we assume that $\Psi_i$ is the shortest one. 
By the claim, the length of $\Psi_{i'}$ is either $h$ or $h+1$, 
if $r_{i'} \in S_i$. 
Let $t^*$ be a time instant that $r_i$ and all $r_{i'} \in S_i$ have 
finished their last cycles in $\Psi$. 
Since $\Omega$ is fair, there is a cycle $\omega = (o,s,f) \in \Omega_i$ 
such that $t^* < o$. 
The state of $r_i$ is $c(h+1)$ at $o$, 
and the state of each $r_{i'} \in S_i$ is either $c_i(h+1)$ or $c_i(h+2)$ at $o$. 
Thus, the state of $r_i$ changes in $\omega$, and hence $\omega \in \Psi_i$.  
It is a contradiction. 

We next show that, for all $k \in {\mathbf N}$, 
if the state of a robot $r_i$ is $Bk$, 
then it will eventually change its state to $R$ in $Y^{(k)}$ 
for any $k \in {\mathbf N}$. 
The proof is by induction on $k$. 

\noindent{\bf Base Case (when $k=1$):}~ 
Let $r_{i'} \in S_i$. 
The state of $r_{i'}$ is either $Y^{(1)}$ or $B^{(1)}$ 
(and not $G^{(1)}$) as long as the state of $r_i$ is $Y^{(1)}$ 
(i.e., either $Bk$, $R$, or $W$). 
Recall the notation 
$\omega_{i}(j) = (o_i(j), s_i(j), f_i(j))$ for all $i$ and $j$. 

By $\mathcal{S}_{VP}$, if $\sigma_i(o_i(j)) = Bk$, 
then it changes its state either $W$ or $R$ in $\omega_i(j)$. 
It changes its state to $W$ if there is an $r_{i'} \in S_i$ 
such that $\sigma_{i'}(o_i(1)) = R$, otherwise 
it changes its state to $R$. 

Suppose that $\sigma_i(o_i(j)) = R$. 
If $\sigma_{i'}(o_i(j)) \in \{R, B, W\}$ for each $r_{i'} \in S_i$, 
since the state of $r_i$ is $Y^{(1)}$ 
then $r_i$ changes its state to $B$ in $\omega_i(j)$. 
Otherwise, if there is an $r_{i'} \in S_i$ such that 
$\sigma_{i'}(o_i(j)) = Bk$, by the observation above, 
$r_{i'}$ changes its state to $R$ or $W$ in $\omega_{i'}(j')$, 
where $j'$ satisfies $o_{i'}(j'-1) < o_i(j) \leq o_{i'}(j')$. 
Furthermore, if $r_{i'}$ changes its state state to $W$, 
it maintains the state as long as the state of $r_i$ is $R$. 
Thus, $r_i$ eventually changes its state to $B$. 

We show that $r_i$ whose state is $Bk$ will eventually change  
its state to $R$ after repeating the loop between $Bk$ and $W$ 
a finite number of times. 
Suppose that $\sigma_i(o_i(j)) = Bk$, and 
let $R_i(j)$ be the set of $r_{i'} \in S_i$ such that 
$\sigma_{i'}(o_i(j)) = R$. 
Robot $r_i$ changes its state to $W$ in $\omega_i(j)$ 
if and only if $R_i(j) = \emptyset$. 
Since $|S_i| < n$, 
$r_i$ will eventually change its state to $R$ after 
repeating the loop at most $n-1$ times, 
since any robot $r_{i'} \in S_i$ with state $B$ will never 
return to $Bk$ as long as the state of $r_i$ is $X^{(1)}$. 

Finally, we show that $r_i$ whose state is $W$ will eventually 
change its state to $Bk$. 
Suppose that $\sigma_i(o_i(j)) = W$. 
If $\sigma_{i'}(o_i(j)) \in \{B, W\}$ for each $r_{i'} \in S_i$, 
then $r_i$ changes its state to $Bk$ in $\omega_i(j)$. 
Otherwise, if there is an $r_{i'} \in S_i$ such that 
$\sigma_{i'}(o_i(j)) \in \{Bk, R\}$, 
it does not change its state in $\omega_i(j)$. 
If $\sigma_{i'}(o_i(j)) = R$, 
then $r_{i'}$ will eventually change its state to $B$. 

If $\sigma_{i'}(o_i(j)) = Bk$, then $r_{i'}$ changes its state to $W$ 
in $\omega_{i'}(j')$, 
where $j'$ satisfies $o_{i'}(j') < o_i(j) \leq o_{i'}(j'+1)$. 
Thus, there is a $r_k \in S_i \cup \{r_i\}$ and $\ell \in {\mathbf N}$ such that 
$o_i(j) < o_k(\ell)$ such that $r_k$ changes its state from 
$W$ to $Bk$ in $\omega_k(\ell)$. 
Since the total number of times that robots other than $r_i$ 
repeat the loops is bounded by $(n-1)^2$, 
$r_i$ will eventually change its state from $W$ to $Bk$. 
The proof of the base case completes. 

\noindent{\bf Induction Step:}~
Suppose that $\sigma_i(o_i(j)) = G^{(k-1)}$ and 
$\sigma_{i'}(o_i(j)) \in \{Bk, G^{(k-1)}\}$ 
for all $r_{i'} \in S_i$. 
Then, $r_i$ changes its state from $G^{k-1}$ to $Bk$ in $\omega_i(j)$. 
The state of $r_i$ is $Bk$ as long as there is an $r_{i'} \in S_i$ 
such that $\sigma_{i'}(o_i(j)) = G^{(k-1)}$. 
We show that $r_{i'}$ will eventually change its state from 
$G^{(k-1)}$ to $Bk$. 
If $r_{i'}$ cannot change its state from $G^{(k-1)}$, 
then there is an $r_{i''} \in S_{i'}$ whose state is neither 
$Bk$ nor $G^{(k-1)}$. 
If the state of $r_{i''}$ is $B^{(k-1)}$, 
then it will eventually change its state to $G^{(k-1)}$. 
If it is $W$ or $R$, we can derive a contradiction, 
because $r_{i''}$ can change its state to $W$ or $R$ when the 
state of $r_i$ is $G^{(k-1)}$. 
Thus, we can apply the proof for the base case to complete the induction step. 
\qed 
\end{proof}

\begin{lemma}
\label{lemma:stationary} 
$\check{F}$ is stationary. 
\end{lemma}
\begin{proof}
A robot $r_i$ can move in $\omega_i(j)$ if the cycle is accepted. 
If $\omega_i(j)$ is accepted, 
the state of $r_i$ is $R$ during the interval $[s_i(j), f_i(j)]$. 
If a cycle $\omega_{i'}(j')$ of a robot $r_{i'} \in S_i$ satisfies 
that $\sigma_i(o_{i'}(j')) = R$, then $\omega_{i'}(j')$ is rejected. 
Thus, $\omega_{i'}(j')$ is accepted, 
only if $o_{i'}(j') \not\in [s_i(j), f_i(j)]$. 
\qed 
\end{proof}

\begin{lemma}
\label{lemma:aligned}
$\check{F}$ is pairwise aligned. 
\end{lemma}
\begin{proof}
Suppose that $\check{F}$ is not pairwise aligned. 
In $\Lambda$, there are cycles $\omega_i(j)$, $\omega_{i'}(j')$, 
and $\omega_{i'}(j'+ \ell)$ satisfying the following conditions: 
\begin{itemize}
\item 
$r_{i'} \in S_i$ (thus, $r_i \in S_{i'}$), 
\item 
$\omega_{i}(j)$ and $\omega_{i'}(j')$ overlap each other, 
\item 
$\omega_i(j)$ and $\omega_{i'}(j'+ \ell)$ overlap each other, 
and 
\item 
$\omega_{i'}(j'+1), \omega_{i'}(j'+2), \ldots, \omega_{i'}(j'+ \ell-1)$ 
are not accepted (hence, they are not the elements of $\Lambda$). 
\end{itemize}
Since $\check{F}$ is stationary, we have 
\begin{equation*}
o_{i'}(j') \leq o_i(j) < s_{i'}(j') < f_{i'}(j') 
< o_{i'}(j' + \ell ) < s_i(j). 
\end{equation*}
Then, for some $0 < k < \ell$, 
there is a cycle $\omega_{i'}(j' + k)$ 
where $r_{i'}$ changes its state from $B$ to $G$. 
It is a contradiction since $\sigma_i(o_{i'}(j')+k) = Bk$. 
\qed 
\end{proof}

Recall that $\Psi_i = (\psi_i(1), \psi_i(2), \ldots) \subseteq \Omega_i$, 
where $\psi_i(j)$ is the cycle of $r_i$ in which it changes its state 
for the $j$th time. 
Scheduler $\mathcal{S}_{VP}$ accepts every cycle $\psi_i \in \Psi_i$ 
in which $r_i$ changes its state from $Bk$ to $R$ and 
accepts no other cycles. 

\begin{proposition}
\label{prop:}
Suppose that $\mathcal{A}$ is vicinity preserving. 
Then if any pair of cycles $\psi_i(j)$ and $\psi_{i'}(j')$ 
in $\Lambda$ 
such that (in $\check{F}$) $r_{i'} \in S_i$ and 
$\psi_i(j) \stackrel{*}{\parallel} \psi_{i'}(j')$ satisfies 
that $\psi_i(j) \parallel \psi_{i'}(j')$, 
then $\check{F}$ is consistent. 
\end{proposition}
\begin{proof}
Since $\mathcal{A}$ is vicinity preserving, 
$r_{i'} \in S_i$ if and only if $r_i \in S_{i'}$  
and $dist(\pi_i(j), \pi_{i'}(j')) > 1$ if $r_{i'} \not\in S_i$. 
\qed 
\end{proof}

\begin{lemma}
\label{lemma:consistent}
$\check{F}$ is consistent. 
\end{lemma}
\begin{proof}
Assume that there are two cycles $\psi_i(j)$ and $\psi_{i'}(j)$ 
in $\Lambda$ such that (in $\check{F}$) 
$\psi_i(j) \stackrel{*}{\parallel} \psi_{i'}(j')$, 
$r_i \in S_{i'}$, and $\psi_i(j) \not\parallel \psi_{i'}(j')$, 
to derive a contradiction. 
Let $\psi_i(j) = (o_i(j), s_i(j), f_i(j))$. 

Since $\psi_i(j) \stackrel{*}{\parallel} \psi_{i'}(j')$, 
there are cycles $\psi_{i_h}(j_h)$ such that 
$\psi_{i_h}(j_h) \parallel \psi_{j_{h+1}}(j_{h+1})$ 
for all $h = 1, 2, \ldots, \ell-1$, 
where $(i_1, j_1) = (i,j)$ and $(i_{\ell}, j_{\ell}) = (i', j')$. 
Note that in cycle $\psi_{i_h}(j_h)$, 
$r_{i_h}$ changes its state from $Bk$ to $R$. 
Recall that in the proof of Lemma~\ref{lemma:fairness}, 
we showed that $\sigma_{i'}(o_i(j))$ is either $c(j)$ or $c(j+1)$. 
Then in cycles $\psi_{i_1}(j_1)$ and $\psi_{i_2}(j_2)$, 
$r_{i_1}$ and $r_{i_2}$ change their states to he same state 
$R$ for the $k$th time for some $k$, 
since $\psi_{i_1}(j_1) \parallel \psi_{i_2}(j_2)$. 
Thus, in $\psi_{i_1}(j_1)$ and $\psi_{i_{\ell}}(j_{\ell})$, 
$r_i$ and $r_{i_{\ell}}$ change their states to $R$ for 
the $k$th time. 

Since $r_{i_{\ell}} \in S_{i_1}$ and 
$\psi_{i_1}(j_1) \not\parallel \psi_{i_{\ell}}(j_{\ell})$, 
we assume that $f_{i_1}(j_1) < o_{i_{\ell}}(j_{\ell})$ 
by the stationarity. 
Since $\psi_{i_{\ell}} \in \Lambda$, 
$\sigma_{i_1}(o_{i_{\ell}}(j_{\ell})) \neq R$, 
which means that there is a cycle $\psi_{i_1}(h_1)$ for 
some $h_1 > j_1$, in which $r_{i_1}$ changes its state from $R$ 
to $B$. 
Obviously, $s_{i_1}(h_1) < o_{i_{\ell}}(j_{\ell})$.

Consider $c_{\ell} = \sigma_{i_{\ell}}(o_{i_1}(h_1))$. 
By definition, $c_{\ell} \in \{W, R, B \}$. 
Observe that $c_{\ell} = R$ means that $c_{\ell} = R^{(k)}$ and 
$c_{\ell} = B$ means that $c_{\ell} = B^{(k)}$. 
Thus, $c_{\ell} = W$, 
because $r_{i_{\ell}}$ changes its state to $R^{(k)}$ in 
$\phi_{i_{\ell}}(j_{\ell})$. 
Then, there is a cycle $\psi_{i_{\ell}}(h_{\ell})$ 
such that $o_{i_1}(h_1) < o_{i_{\ell}}(h_{\ell})$ 
and $h_{\ell} < j_{\ell}$, 
and $r_{i_{\ell}}$ changes its state from $W$ to $Bk$ 
in $\psi_{i_{\ell}}(h_{\ell})$. 

Consider $c_{\ell-1} = \sigma_{i_{\ell-1}}(o_{i_{\ell}}(h_{\ell}))$. 
By the same argument above, $c_{\ell-1} = W$. 
Thus, $c_h = W$ or all $1 \leq h \leq \ell$. 
It is a contradiction, 
since $r_{i_1}$ changes its state to $W$ (in $Y^{(k)}$) from $R^{(k)}$ 
after $\psi_{i_1}(h_1)$. 
\qed 
\end{proof}

\begin{lemma}
\label{lemma:serializable}
$\check{F}$ is serializable. 
\end{lemma}
\begin{proof}
Let $\mathcal{G} = (\{\Lambda_0, \Lambda_1, \ldots, \}, \Rightarrow)$, 
where $\Lambda_0,  \Lambda_1, \ldots$ is the equivalence classes 
of $\Lambda$ with respect to $\stackrel{*}{\parallel}$. 

If two cycles $\omega_i(j)$ and $\omega_{i'}(j')$ are in the same class $\Lambda_{m}$, 
i.e., $\omega_i(j) \stackrel{*}{\parallel} \omega_{i'}(j')$, 
as we showed in the proof of Lemma~\ref{lemma:aligned}, 
$r_i$ and $r_{i'}$ change their states to the same state $R^{(k)}$ 
for some $k$. 

Suppose that $\omega_i(j) \rightarrow \omega_{i'}(j')$. 
Then, in $\omega_i(j)$, $r_i$ changes its state to $R^{(k)}$ for some $k$, 
and in $\omega_{i'}(j')$, $r_{i'}$ changes its state to $R^{(k')}$ 
for some $k'$. 
By definition, $\sigma_i(o_{i'}(j')) \neq R$, which implies $k < k'$. 

If $\check{F}$ is not serializable, there is a loop in $\mathcal{G}$, 
which is a contradiction. 
\qed 
\end{proof}

\begin{lemma}
\label{lemma:natural}
Suppose that $\mathcal{A}$ is vicinity preserving. 
$\check{F}$ is natural. 
\end{lemma}
\begin{proof}
Let $T = (\Lambda_0, \Lambda_1, \Lambda_2, \ldots)$ 
be any topological sort of $\mathcal{G}$. 
Consider any pair of cycles $\psi_i(j)$ and $\psi_{i'}(j')$ in $\Lambda$ 
such that (in $\check{F}$) $k' \leq k < k''$, 
where $\psi_i(j) \in \Lambda_k$, $\psi_{i'}(j'-1) \in \Lambda_{k'}$, 
and $\psi_{i'}(j') \in \Lambda_{k''}$. 
(We assume that $k'=-1$ when $j'=1$ for the consistency.) 

If $r_{i'} \not\in S_i(j)$, then $dist(\pi_i(j), \pi_{i'}(j')) > 0$ 
since $\mathcal{A}$ is vicinity preserving. 

Suppose that $r_{i'} \in S_i(j)$. 
If $o_i(j) \geq o_{i'}(j')$, since $\omega_i(j) \not\parallel \omega_{i'}(j')$, 
$\omega_{i'}(j') \rightarrow \omega_i(j)$, 
which is a contradiction since $k < k'$. 
Thus, $o_i(j) < o_{i'}(j')$. 
\qed 
\end{proof}

By Lemmas~\ref{lemma:fairness}, \ref{lemma:stationary}, 
\ref{lemma:aligned}, \ref{lemma:consistent}, 
\ref{lemma:serializable}, and \ref{lemma:natural}, 
we have the following theorem. 

\begin{theorem}
\label{theorem:SST}
For any vicinity preserving algorithm $\mathcal{A}$ 
for non-luminous SSYNC mobile robots, 
color-based synchronizer $\mathcal{S}_{VP}$ is correct. 
\end{theorem}

\section{Conclusion}
\label{sec:conclusion}

In this paper, we investigated synchronization by ASYNC robots 
with limited visibility. 
We started with a sufficient condition for an ASYNC execution 
to have a similar SSYNC execution. 
Our condition consists of stationarity, 
pairwise alignment, consistency, serializability, and 
naturality on the timing of Look-Compute-Move cycles 
and visibility relation among the robots. 
We then showed the necessity of the five properties 
under a randomized adversary that  selects non-rigid movement and asynchronous observations 
of the robots. 
Our randomized impossibility argument is a novel and stronger 
technique than a worst-case (deterministic) analysis. 
Finally, we presented a color-based synchronizer for luminous ASYNC robots 
together with the limit of color-based synchronizers. 
We showed that there  exists  an algorithm for which 
no color-based synchronizer can guarantee the five properties, 
if the algorithm is not visibility preserving. 
Then, we provided a color-based synchronizer 
that, for a given vicinity preserving algorithm $\mathcal{A}$, 
produces an ASYNC execution that satisfies the 
five properties. 
Thus, luminous ASYNC robots can simulate 
vicinity preserving algorithms designed for (non-luminous) SSYNC robots. 
There are important open problems about a necessary and sufficient condition 
for an algorithm to have a luminous synchronizer. 
The requirement of our color-based synchronizer 
is that an algorithm $\mathcal{A}$ is vicinity preserving. 
It is open whether there exists a color-based synchronizer 
and a general luminous 
synchronizer that works for visibility preserving algorithms. 

\section*{Acknowledgment}

The authors would like to thank Prof. Toshio Nakata for his
precious comments on the infinite product of a Borel 
probability measure space and 
Prof. Giovanni Viglietta for precious discussion in 
University of Ottawa. 

\bibliographystyle{plain}
\bibliography{bibitems}

\begin{thebibliography}{10}

\bibitem{AADFP06}
Dana Angluin, James Aspnes, Zo{\"e} Diamadi, Michael~J. Fischer, and Ren{\'e}
  Peralta.
\newblock Computation in networks of passively mobile finite-state sensors.
\newblock {\em Distributed Computing}, 18(4):235--253, 2006.

\bibitem{CFPS12}
Mark Cieliebak, Paola Flocchini, Giuseppe Prencipe, and Nicola Santoro.
\newblock Distributed computing by mobile robots: gathering.
\newblock {\em SIAM Journal on Computing}, 41(4):829--879, 2012.

\bibitem{DFPSY16}
Shantanu Das, Paola Flocchini, Giuseppe Prencipe, Nicola Santoro, and Masafumi
  Yamashita.
\newblock Autonomous mobile robots with lights.
\newblock {\em Theoretical Computer Science}, 609:171--184, 2016.

\bibitem{DDGRSS14}
Zahra Derakhshandeh, Shlomi Dolev, Robert Gmyr, Andr{\'e}a~W. Richa, Christian
  Scheideler, and Thim Strothmann.
\newblock Brief announcement: Amoebot -- a new model for programmable matter.
\newblock In {\em Proceedings of the 26th ACM Symposium on Parallelism in
  Algorithms and Architectures (SPAA 2014)}, pages 220--222, 2014.

\bibitem{DFSBRS15}
Zahra Derakhshandeh, Robert Gmyr, Thim Strothmann, Rida Bazzi, Andr{\'e}a~W.
  Richa, and Christian Scheideler.
\newblock Leader election and shape formation with self-organizing programmable
  matter.
\newblock In {\em Proceedings of the 21st DNA Computing and Molecular
  Programming}, pages 117--132, 2015.

\bibitem{DFSV18}
Giuseppe~A. Di~Luna, Paola Flocchini, Nicola Santoro, and Giovanni Viglietta.
\newblock Turingmobile: A turing machine of oblivious mobile robots with
  limited visibility and its applications.
\newblock In {\em Proceedings of the 32nd International Symposium on
  Distributed Computing (DISC 2018)}, pages 19:1--19:18, 2018.

\bibitem{DFSVY20}
Giuseppe~A. Di~Luna, Paola Flocchini, Nicola Santoro, Giovanni Viglietta, and
  Yukiko Yamauchi.
\newblock Shape formation by programmable particles.
\newblock {\em Distributed Computing}, 33:69--–101, 2020.

\bibitem{DPV10}
Yoann Dieudonn{\'e}, Franck Petit, and Vincent Villain.
\newblock Leader election problem versus pattern formation problem.
\newblock In {\em Proceedings of the 24th International Symposium on
  Distributed Computing (DISC2010)}, pages 267--281. Springer Berlin
  Heidelberg, 2010.

\bibitem{DSY04a}
Adrian Dumitrescu, Ichiro Suzuki, and Masafumi Yamashita.
\newblock Formations for fast locomotion of metamorphic robotic systems.
\newblock {\em The International Journal of Robotics Research}, 23(6):583--593,
  2004.

\bibitem{DSY04b}
Adrian Dumitrescu, Ichiro Suzuki, and Masafumi Yamashita.
\newblock Motion planning for metamorphic systems: feasibility, decidability,
  and distributed reconfiguration.
\newblock {\em IEEE Transactions on Robotics and Automation}, 20:409--418,
  2004.

\bibitem{FPSW05}
Paola Flocchini, Giuseppe Prencipe, Nicola Santoro, and Peter Widmayer.
\newblock Gathering of asynchronous robots with limited visibility.
\newblock {\em Theoretical Computer Science}, 337:147--168, 2005.

\bibitem{FPSW08}
Paola Flocchini, Giuseppe Prencipe, Nicola Santoro, and Peter Widmayer.
\newblock Arbitrary pattern formation by asynchronous, anonymous, oblivious
  robots.
\newblock {\em Theoretical Computer Science}, 407:412--447, 2008.

\bibitem{FSVY16}
Paola Flocchini, Nicola Santoro, Giovanni Viglietta, and Masafumi Yamashita.
\newblock Rendezvous with constant memory.
\newblock {\em Theoretical Computer Science}, 621:57–--72, 2016.

\bibitem{FSW19}
Paola Flocchini, Nicola Santoro, and Koichi Wada.
\newblock On memory, communication, and synchronous schedulers when moving and
  computing.
\newblock In {\em Proceedings of the 23rd International Conference on
  Principles of Distributed Systems (OPODIS 2019)}, pages 25:1--25:17, 2020.

\bibitem{FYOKY15}
Nao Fujinaga, Yukiko Yamauchi, Hirotaka Ono, Shuji Kijima, and Masafumi
  Yamashita.
\newblock Pattern formation by oblivious asynchronous mobile robots.
\newblock {\em SIAM Journal on Computing}, 44(3):740--785, 2015.

\bibitem{FYOKY17}
Nao Fujinaga, Yukiko Yamauchi, Hirotaka Ono, Shuji Kijima, and Masafumi
  Yamashita.
\newblock Erratum: Pattern formation by oblivious asynchronous mobile robots,
  2017.
\newblock http://tcs.inf.kyushu-u.ac.jp/~yamauchi/publications.html.

\bibitem{SY99}
Ichiro Suzuki and Masafumi Yamashita.
\newblock Distributed anonymous mobile robots: Formation of geometric patterns.
\newblock {\em SIAM Journal on Computing}, 28(4):1347--1363, 1999.

\bibitem{T11}
Terence Tao.
\newblock {\em An Introduction to Measure Theory}.
\newblock American Mathematical Society, 2011.

\bibitem{YS10}
Masafumi Yamashita and Ichiro Suzuki.
\newblock Characterizing geometric patterns formable by oblivious anonymous
  mobile robots.
\newblock {\em Theoretical Computer Science}, 411:2433--2453, 2010.

\bibitem{YY13}
Yukiko Yamauchi and Masafumi Yamashita.
\newblock Pattern formation by mobile robots with limited visibility.
\newblock In {\em Proceedings of the 20th International Colloquium on
  Structural Information and Communication Complexity (SIROCCO 2013)}, pages
  201--212. Springer International Publishing, 2013.

\end{thebibliography}

\end{document}